\documentclass[journal]{IEEEtran}
\IEEEoverridecommandlockouts

\usepackage[cmex10]{amsmath}
\allowdisplaybreaks[4]

\usepackage{enumitem}
\usepackage{graphicx}
\usepackage{color}
\usepackage{amsmath}
\usepackage{mathtools}
\usepackage{multicol}
\usepackage{multirow}
\usepackage[english]{babel}
\usepackage{blindtext}
\usepackage{algorithm}
\usepackage{algorithmic}
\usepackage{balance}
\usepackage{amsfonts}
\usepackage{bm}
\usepackage{stfloats}
\usepackage{subfig}
\usepackage{amsthm}
\usepackage{amssymb}
\usepackage{setspace}
\usepackage[nosort]{cite}
\usepackage{CJK}
\usepackage{cite}
\usepackage{caption}
\usepackage{comment}
\usepackage{array,multirow}
\usepackage{graphicx}
\newtheorem{lemma}{Lemma}
\newtheorem{theorem}{Theorem}
\newtheorem{corollary}{Corollary}

\theoremstyle{plain}

\usepackage[table]{xcolor}

\begin{document}
\captionsetup[figure]{labelformat={default},labelsep=period,name={Fig.}}

\title{Coverage and Spectral Efficiency of NOMA-Enabled LEO Satellite Networks with Ordering Schemes}

\author{Xiangyu Li,
        Bodong Shang,~\IEEEmembership{Member,~IEEE},
        Qingqing Wu,~\IEEEmembership{Senior Member,~IEEE},
        Chao Ren,~\IEEEmembership{Member,~IEEE}
\thanks{This work was supported in part by the Ningbo Natural Science Foundation under grant 20253021, in part by the State Key Laboratory of Integrated Services Networks, Xidian University, and in part by the open research fund of National Mobile Communications Research Laboratory, Southeast University (No. 2026D11). The work of Qingqing Wu is supported by Shanghai JXW Fund JJ-GGFWPT-01-24-0030. 
\textit{(Corresponding author: Bodong Shang)}}
\thanks{Copyright (c) 20xx IEEE. Personal use of this material is permitted. However, permission to use this material for any other purposes must be obtained from the IEEE by sending a request to pubs-permissions@ieee.org.}
\thanks{Xiangyu Li is with the Department of Electronic Engineering, Shanghai Jiao Tong University, Shanghai 200240, China, and is also with the State Key Laboratory of Integrated Services Networks, Xidian University, Xi’an, China, and is also with Zhejiang Key Laboratory of Industrial Intelligence and Digital Twin, Eastern Institute of Technology, Ningbo, Zhejiang 315200, China (e-mail: xyli@eitech.edu.cn).}
\thanks{Bodong Shang is with State Key Laboratory of Integrated Services Networks, Xidian University, Xi’an, China, and is also with National Mobile Communications Research Laboratory, Southeast University, Nanjing, China, and is also with Zhejiang Key Laboratory of Industrial Intelligence and Digital Twin, Eastern Institute of Technology, Ningbo, Zhejiang 315200, China (e-mail: bdshang@eitech.edu.cn).}
\thanks{Qingqing Wu is with the Department of Electronic Engineering, Shanghai Jiao Tong University, Shanghai 200240, China (e-mail: qingqingwu@sjtu.edu.cn).}
\thanks{Chao Ren is with University of Science and Technology Beijing, Beijing 100083, China (e-mail: chaoren@ustb.edu).}
}

\maketitle

\begin{abstract}
This paper investigates an analytical model for low-earth orbit (LEO) multi-satellite downlink non-orthogonal multiple access (NOMA) networks. The satellites transmit data to multiple NOMA user terminals (UTs), each employing successive interference cancellation (SIC) for decoding. Two ordering schemes are adopted for NOMA-enabled LEO satellite networks, i.e., mean signal power (MSP)-based ordering and instantaneous signal-to-inter-satellite-interference-plus-noise ratio (ISINR)-based ordering. For each ordering scheme, we derive the analytical expression for the coverage probability of each typical UT. Moreover, we discuss how coverage is influenced by SIC, main-lobe gain, and tradeoffs between the number of satellites and their altitudes. Additionally, two user fairness-based power allocation (PA) schemes are considered, and PA coefficients with the optimal number of UTs that maximize their sum spectral efficiency (SE) are studied. Simulation results show that there exists a maximum effective signal-to-inter-satellite-interference-plus-noise ratio (SINR) threshold for each PA scheme that ensures the operation of NOMA in LEO satellite networks, and NOMA provides performance gains only when the target SINR is below a certain threshold. Compared with orthogonal multiple access (OMA), NOMA increases UTs' sum SE by as much as 35\%.
Furthermore, for most SINR thresholds, the sum SE increases with the number of UTs to the highest value, whilst the maximum sum SE is obtained when there are two UTs.
\end{abstract}

\begin{IEEEkeywords}
Low-earth orbit satellite networks, non-orthogonal multiple access, coverage probability, spectral efficiency, interference modeling.
\end{IEEEkeywords}

\IEEEpeerreviewmaketitle

\section{Introduction}
Recent decades have witnessed an overwhelming trend in the development of satellite networks~\cite{kodheli2020satellite}. Compared to terrestrial networks (TNs), satellite networks are envisioned to have a more significant potential for universal coverage and ubiquitous connectivity feasibly and economically, particularly for unserved or underserved regions~\cite{okati2022nonhomogeneous}.
Different space-borne platforms, including geostationary-earth orbit (GEO) satellites, medium-earth orbit (MEO) satellites, and low-earth orbit (LEO) satellites, have been established. Among them, LEO satellites are attracting the attention of researchers from academia and industry for their higher data rates, lower latency, reduced power consumption, and decreased production and launch costs~\cite{li2024bistatic,li2025instant}.
This trend has contributed to the explosive growth of various LEO satellite-based applications and services, which demand higher data rates and spectral efficiency (SE) for an increasing number of user terminals (UTs)~\cite{yan2023noma,ouyang2025vulnerability}.

Despite the mounting number of LEO satellites in constellations such as Starlink, OneWeb, and Telesat, the number of UTs keeps increasing. 
This means that the limited spectrum resources deployed by satellite networks may not be enough for all these UTs, and the quality of service (QoS) requirements of resource-intensive applications and services can hardly be met. 
In view of this situation, there is an urgent need to develop innovative technologies and efficient resource allocation strategies, such as multiple access techniques, to improve spectrum utilization and support simultaneous communications with a growing number of UTs.

\subsection{Related Works}
Many recent works have studied LEO satellite networks \cite{okati2020downlink,park2022tractable,jia2022analytic,deng2021ultra}.
Okati \textit{et al.} \cite{okati2020downlink} proposed a basic satellite framework under Rayleigh fading to pave the way for accurate analysis and design of future dense satellite networks. 
To capture the effects of relevant parameters, such as satellite altitude, density, and fading parameters, a more tractable and practical model was developed in \cite{park2022tractable} under Nakagami fading, showing that there is an optimal number of satellites at different satellite heights to maximize coverage performance. 
To provide deployment guidance for mega-satellite constellations, authors in \cite{jia2022analytic} investigated an analytical approach to evaluate the uplink performance under a Nakagami fading-approximated Shadowed-Rician fading model. 
Ultra-dense LEO satellite constellations were also investigated in \cite{deng2021ultra} to show that the proposed LEO satellite constellation has a higher coverage ratio than the benchmark Telesat constellation.
However, these works assumed that a satellite serves only one UT within a given time-frequency resource block (RB), an approach that becomes increasingly impractical as the number of ground UTs continues to grow.
Therefore, it is imperative to explore and adopt mature and advanced methodologies to effectively address the challenges posed by scarce spectrum resources and maximize their utilization.

To share these limited spectrum resources, current satellite systems mainly adopt the orthogonal multiple access (OMA) schemes, which consist of time-division multiple access (TDMA), frequency-division multiple access (FDMA), code-division multiple access, space-division multiple access (SDMA), etc.~\cite{ding2017survey}. Unfortunately, although the OMA scheme effectively avoids interference among UTs by serving only one UT per time-frequency RB, improvements in resource utilization efficiency are limited. In addition, a higher service priority is assigned to UTs with better channel conditions,
which sacrifices the fairness of resource allocation among them.
To address such concerns, power-domain non-orthogonal multiple access (NOMA) is deemed a promising candidate for more efficient spectrum reuse. With NOMA, multiple signals are superposed in the power domain before transmission; the receiver applies successive interference cancellation (SIC) to decode the superposed signals and remove the interference caused by superposition~\cite{liu2017nonorthogonal}.
Thus, multiple UTs can be served simultaneously in the same time-frequency RB, and SE can be improved at the cost of reasonably increased complexity~\cite{saito2013non}.
At the same time, the reliability and fairness of the UT can be improved.

In recent years, although many works have explored NOMA-enabled TNs, their extensions to NOMA-enabled non-terrestrial networks (NTNs), especially LEO satellite networks, remain limited and lack comprehensive investigation. 
Ding et al.~\cite{ding2014performance} were among the first to investigate the outage performance and the ergodic sum rate of the NOMA downlink system with randomly distributed UTs. NOMA uplink transmission was proposed in~\cite{zhang2016uplink}, where the outage probabilities and data rates of the first and second UTs were compared with closed-form solutions. However, these studies primarily focus on TNs and overlook the complex interference scenarios in NTNs, including inter-satellite interference. 
While~\cite{liang2019non} and~\cite{ali2019downlink} investigated the uplink and downlink performance of NOMA in Poisson networks with a general number of UTs with both inter-cell and intra-cell interference, they lack sufficient channel modeling and analysis suitable for cross-satellite interference.
Further, power allocation (PA) is another critical challenge in NOMA systems and has been widely studied. For example, the authors in~\cite{cui2016novel} showed a PA scheme based on maximum-minimum rate fairness to ensure QoS for the worst UT. In~\cite{shipon2016dynamic}, a joint design of dynamic user clustering and PA was investigated for NOMA uplink and downlink, with closed-form expressions provided for optimal PAs in any cluster size. 
However, these studies did not fully account for the characteristics of satellite communication systems, particularly the resource competition among NOMA UTs in NTNs.

While the potential benefits of NOMA in TNs have been widely discussed, recent studies have shown that it also holds great promise for satellite communications. In~\cite{yan2019application}, a generalized framework for the applicability of NOMA in satellite networks was introduced, highlighting its potential advantages in ergodic rate enhancement~\cite{yan2019ergodic}, energy efficiency maximization~\cite{mirbolouk2022relay}, and outage probability reduction~\cite{gao2020performance}. 
However, most of these studies were conducted under specific conditions, such as assuming ideal SIC or neglecting the complexity of cross-satellite interference.
In~\cite{yan2023noma}, to guarantee individual performance, a QoS-limited satellite NOMA system was studied to show that a relatively large transmit signal-to-noise ratio (SNR) was suitable for UTs with better channel conditions, while UTs with poorer link gain desired a relatively small SNR. 
In the context of NOMA-enabled land mobile satellite networks, the uplink and downlink outage probabilities were investigated in~\cite{tegos2020outage} and~\cite{yan2018performance}, respectively.
In~\cite{li2024performance,yang2025performance}, an OMA / NOMA-aided satellite communication network with a near-user and a far-user was studied to show the superiority of NOMA over OMA. 
Nonetheless, the impact of inter-satellite interference and resource competition, which introduce more complex interference patterns, was not adequately considered. In addition, the performance of integrating more than two UTs remains unknown.

To systematically capture the irregular topology of LEO satellite networks, stochastic geometry has been employed to provide a unified mathematical paradigm~\cite{haenggi2009stochastic}. Several works have focused on the performance analysis of LEO satellite networks using this method. For example,~\cite{talgat2020stochastic} analyzed the downlink coverage probability considering satellite altitudes and numbers, where satellite locations were distributed according to the binomial point process (BPP) in a finite space. However, to improve tractability, the Poisson point process (PPP) and the spherical Poisson point process (SPPP) have become the preferred choice for modeling LEO satellite networks~\cite{okati2021modeling}. They have been used to study the outage probability, ergodic capacity, and uplink interference~\cite{jia2021uplink}, and have been extended to scenarios involving coordinated beamforming~\cite{kim2024coverage}, joint transmission~\cite{li2024analytical}, multi-connectivity between TNs and NTNs~\cite{shang2024multi}, etc.
Despite these advances, the joint effects of line-of-sight (LoS) conditions, satellite numbers and altitudes, imperfect SIC, main-lobe gains, PA schemes, and cross-satellite interference have not been sufficiently studied for downlink satellite NOMA networks.

\subsection{Motivations}
We have noticed several research gaps in NOMA-enabled LEO multi-satellite networks. 
\begin{itemize}
    \item First, previous studies of satellite NOMA transmissions \cite{li2024performance,yang2025performance} typically considered a two-user scenario comprising a near-user and a far-user. However, as the number of users requiring satellite services continues to grow, the impact of the NOMA mechanism with multiple users on system performance remains unknown.
    \item Second, when applying NOMA, previous works assumed that there is only one satellite, i.e., the serving satellite, thereby neglecting the influence of inter-satellite interference. Furthermore, the impacts of other system parameters, such as practical SIC effects, main-lobe and side-lobe gains, number of satellites, and their altitude, have not been explored.
    \item Third, it also remains uncertain whether a channel-dependent ordering scheme should still be employed or if a distance-dependent ordering scheme might be more appropriate for simplicity, given that the multi-path fading components are relatively weaker than the direct propagation paths. 
    \item Fourth, the implementation of user fairness through various PA schemes requires investigation.
\end{itemize}

Motivated by the above observations, this paper aims to investigate the potential of NOMA from a system-level perspective to improve the performance of satellite communication networks. 
It is also essential to rigorously quantify the extent of performance improvement and systematically examine the influence of various parameters.
In particular, stochastic geometry is applied to investigate a large-scale LEO multi-satellite downlink NOMA network that incorporates both intra- and inter-satellite interference.
We also study the coverage and sum SE of NOMA UTs in terms of mean signal power (MSP) ordering, which is tantamount to distance-based ordering, and instantaneous signal-to-inter-satellite-interference-plus-noise ratio (ISINR) ordering.

\subsection{Contributions and Paper Organization}
The main contributions of this paper are as follows:
\begin{itemize}
    \item \textbf{LEO Satellite NOMA System Design:} We establish a NOMA-enabled LEO multi-satellite network, where a typical LEO satellite simultaneously serves multiple UTs within its serving area in the same time-frequency RB. To capture the randomness of both satellites and UTs, we model the satellites in an SPPP manner at the same orbital altitude and let the UTs follow a uniform distribution. We also compare the PPP constellation model with the Starlink constellation model to assess its accuracy and effectiveness.
    \item \textbf{Modeling and Analysis:} To capture the impact of LoS effects on system performance, we use the Nakagami-$m$ distribution to model the small-scale fading channel between satellites and UTs. The closed-form expressions for the coverage probabilities of typical UTs are derived under two types of ordering schemes using stochastic geometry and order statistics. By conducting simulations and numerical calculations, we quantitatively evaluate the system in terms of coverage probability and SE.
    \item \textbf{System Design Insights:} Simulations and numerical results show that strong LoS effects, better SIC in UTs, and higher main-lobe gains of the serving satellite positively contribute to the coverage improvements of UTs. There is a trade-off between the number of satellites and their altitudes to maximize the mean coverage, and consequently, the mean data rate for the served UTs. We also show that every PA scheme has a maximum effective SINR threshold that ensures NOMA operation, and that the benefits of NOMA over OMA are restricted to certain SINR thresholds. The sum SE increases with the number of UTs, but reaches a specific highest value; the maximum sum SE is obtained when only two UTs are served.
\end{itemize}

The remainder of this paper is organized as follows. Section \uppercase\expandafter{\romannumeral2} elaborates the model of the LEO satellite NOMA system.
In Section \uppercase\expandafter{\romannumeral3}, the coverage probabilities of typical UTs are derived under two types of ordering schemes.
Section \uppercase\expandafter{\romannumeral4} provides closed-form expressions of coverage probability under particular channel coefficients.
In Section \uppercase\expandafter{\romannumeral5}, simulations and numerical results are presented to show the accuracy of the derivations and to evaluate the impacts of related parameters.
Finally, Section \uppercase\expandafter{\romannumeral6} concludes the paper. The main notations and related descriptions are summarized in Table \ref{notation}.

\begin{table}[!t]
\setlength{\abovecaptionskip}{0cm}
\captionsetup{font={scriptsize}}
\caption{Main notations and descriptions}
\label{notation}
\begin{center}
\small
\begin{tabular}{c|p{6.4cm}}
\hline
\hline
\textbf{Notations } &  \textbf{Descriptions} \\ 
\hline
\hline
$l_i$ & Distance from the typical satellite to UT$_i$ \\
$l_0$ & Path-loss on the reference distance \\
$\alpha$ & Path-loss coefficient \\
$Z_i$ & ISINR of UT$_i$ \\
$N_{\textcolor{black}{\text{S}}}$ & Total number of satellites\\
$R_{\textcolor{black}{\text{T}}}$ & Radius of the typical satellite's serving area\\
$R_{\textcolor{black}{\text{E}}}$ & Radius of the earth\\
$H_{\textcolor{black}{\text{S}}}$ & Altitude of satellites \\
$R_{\textcolor{black}{\text{S}}}$ & Revolution radius of satellites\\
$\lambda_{\textcolor{black}{\text{S}}}$ & Density of satellites \\
$N_{\textcolor{black}{\text{U}}}$ & Total number of typical UTs\\
$O$ & Projection point of the typical satellite on the earth \\
$\mathbf{s}_t$ & A symbol for location of LEO satellite, where $t=0$ for the typical satellite, $t=1$ for nearest interfering satellite, and $t\geq 2$ for remaining interfering satellites likewise \\
$\mathbf{u}_i$ & Location of UT$_i$ \\
$\left\|\mathbf{d}_{i}\right\|$ & Distance between satellite $\mathbf{s}_t$ and UT$_i$ \\
$c$ & Speed of light \\
$q_j$ & Transmitted data symbol \\
$f_c$ & Carrier frequency \\
$\hat{h}_i$ & Channel coefficient between typical satellite and UT$_i$ \\
$h_i$ & Small-scale fading parameter between typical satellite and UT$_i$ \\
$\hat{g}_i$ & Channel coefficient between interfering satellite and UT$_i$ \\
${g}_{\mathbf{d}_{i}}$ & Small-scale fading parameter between interfering satellites and UT$_i$ \\
$m$ & Shape parameter of Nakagami distribution  \\
$\Omega$ & Scale parameter of Nakagami distribution \\
$\kappa$ & Shape parameter of Gamma distribution  \\
$\beta$ & Inverse scale parameter of Gamma distribution \\
$\varpi$ & Residual intra-satellite interference (RI) factor\\
$P$ & Total transmit power of a satellite \\
$p_i$ & Fraction of power allocated for UT$_i$ \\
$G_\text{ml}$ & Main-lobe gain of a satellite \\
$G_\text{sl}$ & Side-lobe gain of a satellite \\
$\omega$ & Addictive white Gaussian noise (AWGN)\\
$I_{j,i}^{\text{intra}}$ & Intra-satellite interference signal power when UT$_i$ decodes the message for UT$_j$\\
$I_{i}^{\text{inter}}$ & Inter-satellite interference signal power received by UT$_i$ \\
$\Phi_{\textcolor{black}{\text{S}}}$ & Set of all interfering satellites \\
\hline
\hline
\end{tabular}
\vspace{-2.0em} 
\end{center}
\end{table}

\section{System Model}
In this section, we describe the considered LEO satellite NOMA system, which consists of the network model, channel model, and signal analysis. Specifically, for $1\leq i\leq N_{\textcolor{black}{\text{U}}}$, where $N_{\textcolor{black}{\text{U}}}$ denotes the total number of typical UTs, two NOMA ordering schemes are considered:
\begin{itemize}
    \item \textbf{Ordering scheme 1: MSP ordering} is a fading-free ordering that uses the total unit power received at UTs. For the $i$-th nearest UT from the typical satellite, the typical UTs are indexed according to their ascending distance $l_i$ from the typical satellite, i.e., descending ${l_i}^{-\alpha}$, where $\alpha$ is the path-loss coefficient.
    \item \textbf{Ordering scheme 2: ISINR ordering} is a fading-based ordering that involves ISINR ${{Z}_{i}}=\frac{{{l}_{i}}^{-\alpha }{{\left| h \right|}^{2}}}{I_{i}^{\text{inter}}+{\bar{\sigma}^{2}}}$ of the $i$-th typical UT. In ISINR ordering, UT$_i$ has the $i$-th largest ISINR, and typical UTs are indexed in terms of descending ordered $Z_i$.
\end{itemize}

\subsection{Network Model}

\captionsetup{font={scriptsize}}
\begin{figure*}[!t]
\begin{center}
\setlength{\abovecaptionskip}{+0.2cm}
\setlength{\belowcaptionskip}{-0.0cm}
\centering
  \includegraphics[width=7.0in, height=3.0in]{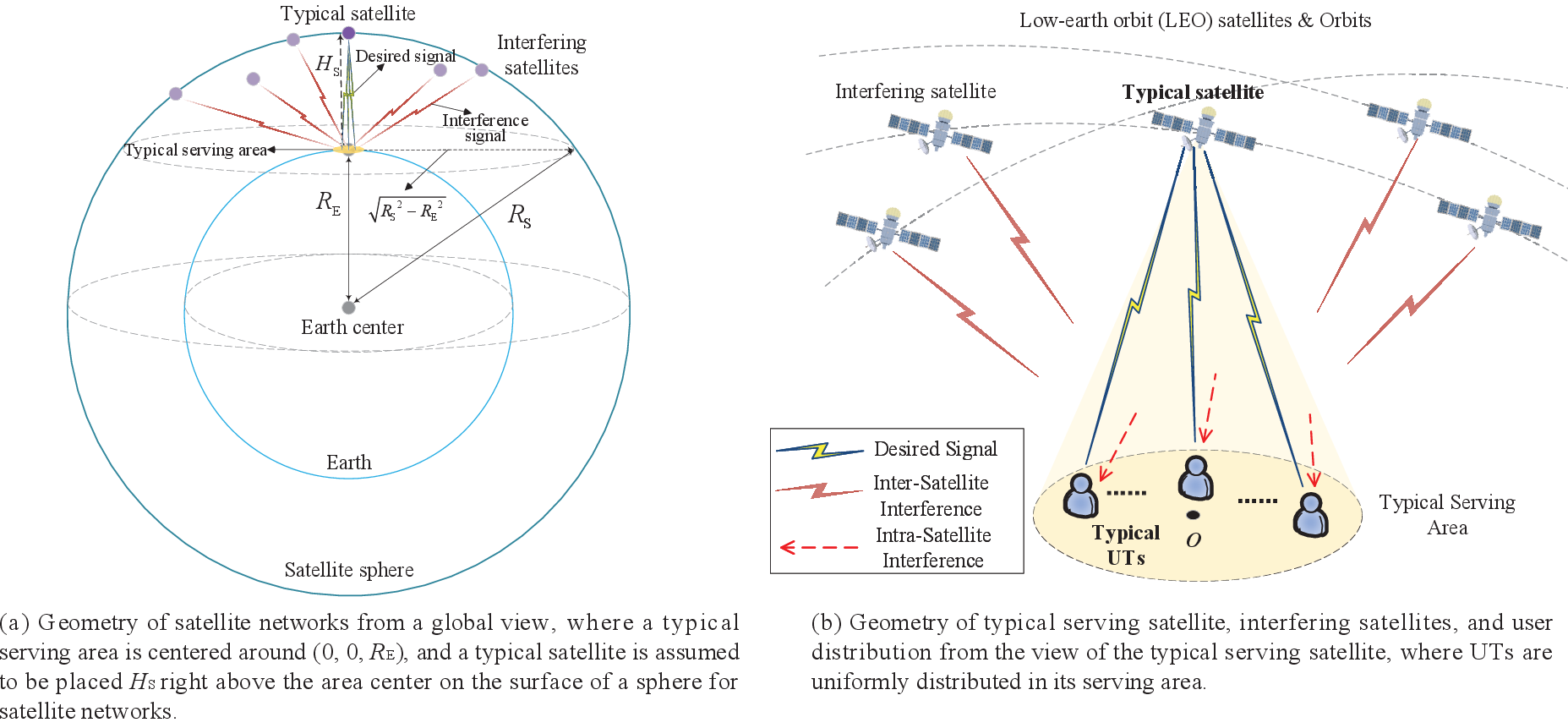}
\renewcommand\figurename{Fig.}
\caption{\scriptsize \textcolor{black}{An illustration of NOMA-enabled LEO multi-satellite networks.}}
\label{fig1}
\end{center}
\vspace{-6mm}
\end{figure*}

As shown in Fig. \ref{fig1}(a), a NOMA-enabled LEO multi-satellite network consists of $N_{\textcolor{black}{\text{S}}}$ LEO satellites at an altitude $H_{\textcolor{black}{\text{S}}}$ on the surface of a shell sphere. Denote the radius of the earth as $R_{\textcolor{black}{\text{E}}}$, and the revolution radius of the satellite, i.e., the distance from a satellite to the center of the earth, as $N_{\textcolor{black}{\text{S}}} = R_{\textcolor{black}{\text{E}}} + H_{\textcolor{black}{\text{S}}}$. LEO satellites are distributed according to a homogeneous SPPP with a density 
$\lambda_{\textcolor{black}{\text{S}}} = \frac{N_{\textcolor{black}{\text{S}}}}{4\pi \left(R_{\textcolor{black}{\text{E}}} + H_{\textcolor{black}{\text{S}}}\right)^2}$. 

In Fig. \ref{fig1}(b), it is assumed that a typical satellite is located at the top of the spherical shell surface.
A total of $N_{\textcolor{black}{\text{U}}}$ UTs served by the satellite in one frequency band are uniformly distributed in a spot beam, i.e., a circular serving area\footnote{\textcolor{black}{Considering the much larger Earth radius $R_{\textcolor{black}{\text{E}}} = 6,371.393$ km and the satellite altitude $H_{\textcolor{black}{\text{S}}} = 500$ km to be mentioned in the following sections, the Earth curvature effects on the projected ground area can be negligible. Therefore, the terrestrial region of interest can be approximated as a circular planar disk for analytical tractability, following the modeling approach in \cite{dong2024stochastic}. More detailed curvature-aware modeling can be found in \cite{dong2025modeling} and is left for future work.}}, of the satellite. 
The service area is centered at the vertical projection point $O$ of the typical satellite on the ground, within which the UTs are referred to as typical UTs.
It is assumed that UTs remain quasi-static during the considered transmission period. This ``snapshot''-based analysis enables a clear evaluation of the fundamental effects of ordering schemes and interference at the system level, while avoiding the additional complexity introduced by user mobility and time-varying satellite geometry.
By excluding the cases where the typical satellite is not the nearest satellite to typical UTs, it can be assumed that the serving area of satellites does not overlap, which is consistent with the fact that typical UTs are only served by the typical satellite within a time-frequency RB. 
The frequency reuse scheme is implemented to allow UTs outside the service area of a satellite to be served by other satellites or frequencies.

\subsection{Channel Model}
Let $\mathbf{s}_t$ with integer $t\ge 0$ represent the locations of LEO satellites,
where ${\mathbf{s}_{0}}$ implies the typical satellite, ${\mathbf{s}_{1}}$ implies the nearest interfering satellite and ${\mathbf{s}_{t}}$ for $t\geq 2$ implies the remaining interfering satellites.
The large-scale fading is modeled as $L(l_i) = \left( \frac{c}{4\pi f_c} \right)^2 {l_i}^{-\alpha} = L_{\text{pl}} {l_i}^{-\alpha}$~\cite{jia2022analytic}, where $c$ is the speed of light, $f_c$ is the carrier frequency, $\alpha$ is the path-loss coefficient, and $l_i$ is the distance between the satellite ${\mathbf{s}_{0}}$ and the typical UT$_i$.
The channel coefficient is then denoted as ${\hat{h}_i=l_{i}}^{-\alpha/2}h_i$, where $h_i$ is the small-scale fading coefficient.
Similarly, the channel coefficient between an interfering satellite and UT$_i$ is denoted as ${{\hat{g}}_{i}}={\left\| {\mathbf{d}_{i}} \right\|}^{-\alpha /2} {g}_{\mathbf{d}_{i}}$, where $\left\|\mathbf{d}_{i}\right\|=\left\|\mathbf{s}_{t}-\mathbf{u}_{i}\right\|$, $t\geq 1$, $\mathbf{u}_{i}$ is the location of UT$_i$, and ${g}_{\mathbf{d}_{i}}$ is the small-scale fading coefficient. 

The Nakagami-$m$ distribution is adopted for the modeling of small-scale channel fading
because it is versatile for various small-scale fading conditions, e.g., $m=1$ for the Rayleigh channel and $m=\frac{(K+1)^2}{2K+1}$ for the Rician channel \cite{park2022tractable}. 
Assume that $\Omega=\mathbb{E}\left\{\left|h_i\right|^2\right\}=\mathbb{E}\left\{\left|g_{\mathbf{d}_{i}}\right|^2\right\}=1$. The probability density function (PDF) of $\left|h_i\right|\left(\left|g_{\mathbf{d}_{i}}\right|\right)$ is given by
\begin{equation}
    f_{\left|h_i\right|\left(\left|g_{\mathbf{d}_{i}}\right|\right)}(x) 
    = \frac{2m^m}{\Gamma(m)\Omega^m}x^{2m-1}e^{-\frac{m}{\Omega} x^2}, 
\end{equation}
for $x \geq 0$, and Gamma function $\Gamma(m)=(m-1)!$ for integer $m>0$.
The change in parameter $m$ makes it suitable for modeling signal propagation conditions from severe to moderate\footnote{\textcolor{black}{Compared with the Nakagami-$m$ fading model, the Shadowed-Rician (SR) fading has been recognized as a more suitable statistical representation of satellite channel characteristics. Nevertheless, the inherent complexity of its PDF renders the derivation of closed-form analytical expressions highly challenging, particularly in the presence of interference. Moreover, as demonstrated in \cite{abdi2003new}, the SR distribution can be approximated by a Gamma random variable by applying the moment matching technique, and this has already been used in \cite{jia2022analytic,jia2021uplink}. Therefore, although adopting Nakagami-$m$ fading is not the most precise choice, it allows for a more tractable analysis and facilitates the derivation of insightful theoretical results \cite{yang2025performance}.}}.
Correspondingly, the power of small-scale Nakagami-$m$ channel coefficient $\left|h_i\right|\left(\left|g_{\mathbf{d}_{i}}\right|\right)$ subject to Gamma distribution whose PDF is 
\begin{equation}
    {f_{{{\left| h_i \right|}^{2}}\left( {{\left| g_{\mathbf{d}_{i}} \right|}^{2}} \right)}}\left( x \right) 
    \approx \frac{{{x}^{\kappa -1}}{{e}^{-\beta x}}{{\beta }^{\kappa }}}{\Gamma \left( \kappa  \right)}, 
\end{equation}
where $\kappa$ is a shape parameter and $\beta$ is an inverse scale parameter. Due to their square-law relationship~\cite{nakagami1960the,holm2004sum}, we have $m=\kappa=\beta$. To capture their exact position, $m=\kappa$ and $\beta$ are written separately.

\subsection{Signal Analysis}
Apart from the desired signal from the typical satellite, two other types of interference signals also count. Firstly, the typical UTs suffer from inter-satellite interference, which arises from signals transmitted by non-serving interfering satellites.

Secondly, intra-satellite interference is caused by the co-channel interference of the NOMA scheme. 
\textcolor{black}{
The core principle of SIC lies in a decoding–reconstruction–subtraction (DRS) process. Upon receiving the composite signal, SIC first decodes the strongest user while treating others as interference. The decoded data are then re-encoded according to the estimated channel and modulation parameters to reconstruct the signal, which is subsequently subtracted from the aggregate signal. This subtraction reduces interference, enabling the following DRS stages to decode the remaining users more effectively.
Note that the reconstructed signal tends to be very close to the received signal under perfect conditions. However, since SIC is generally imperfect in practice, the residual intra-satellite interference (RI) factor $\varpi \in [0,1]$, i.e., the error propagation factor, to show the decoding conditions similar to \cite{sun2016non}. This factor $\varpi$ is adopted to measure its impact on the system.
}
Herein, $\varpi=0$ represents the perfect SIC condition and $\varpi=1$ corresponds to no SIC at all.

The total transmit power of a satellite is denoted as $P$, and the fraction of power allocated for UT$_i$ as $p_i \in (0,1)$ so that $\sum_{i=1}^{N_{\textcolor{black}{\text{U}}}} p_i = 1$. The main-lobe gain and the side-lobe gain of the satellite are set to $G_\text{ml}$ and $G_\text{sl}$, respectively\footnote{It is important to note that the main lobe and side lobes are non-overlapping in the angular domain, which is typically ensured through careful satellite beam design and isolation techniques. The following  (\ref{Formula_received_signal_y}) models the interference power contribution from the side lobes associated with interfering satellites rather than from the serving satellite itself \cite{park2022tractable}.}. 
Denote the set of all interfering satellites by $\Phi_{\textcolor{black}{\text{S}}}$. 
Then, in the serving area of the typical satellite, the received signal at UT$_i$ of the message intended for UT$_j$ for $i\leq j\leq {\textcolor{black}{N_\text{U}}}$ is expressed as
\begin{align}
    y_{j}^{i}
    = & \sqrt{{{G}_{\text{ml}}}}\sqrt{{{p}_{j}}P} L_{\text{pl}} {{l}_{i}}^{-\alpha /2}{{h}_{i}} {{q}_{j}} \nonumber \\
    & +\sqrt{{{G}_{\text{ml}}}}\left( \sum\limits_{m=1}^{j-1}{\sqrt{{{p}_{m}}P}}+\varpi \sum\limits_{k=j+1}^{N_{\textcolor{black}{\text{U}}}}{\sqrt{{{p}_{k}}P}} \right) L_{\text{pl}} {{l}_{i}}^{-\alpha /2}{{h}_{i}} {{q}_{j}} \nonumber \\
    & + \sqrt{{{G}_{\text{sl}}}}\sqrt{P} {{\sum\limits_{\mathbf{s}\in \Phi_{\textcolor{black}{\text{S}}} } L_{\text{pl}} {\left\| {\mathbf{d}_{i}} \right\|}}^{-\alpha /2}} g_{\mathbf{d}_{i}} {{q}_{j}} + \omega,
\label{Formula_received_signal_y}
\end{align}
where $q_j$ is the transmitted data symbol with $|q_j|^2=1$, $\omega$ is the additive white Gaussian noise (AWGN) and $\omega \sim \mathcal{CN}(0,\sigma^2)$.

Then, the SINR at UT$_i$ of the message intended for UT$_j$ is given by (\ref{Formula_SINR}) at the top of the next page, 
\begin{figure*}[!ht]
\setlength{\abovecaptionskip}{-1.5cm}
\setlength{\belowcaptionskip}{-0.5cm}
\normalsize
\vspace{-2mm}
\begin{equation}
\begin{aligned}
    \text{SINR}_{j}^{i}
     & = \frac{{{p}_{j}}{{l}_{i}}^{-\alpha }{{\left| {{h}_{i}} \right|}^{2}}}{\left( \sum\limits_{m=1}^{j-1}{{{p}_{m}}}+\varpi \sum\limits_{k=j+1}^{N_{\textcolor{black}{\text{U}}}}{{{p}_{k}}} \right){{l}_{i}}^{-\alpha }{{\left| {{h}_{i}} \right|}^{2}}+\frac{{{G}_{\text{sl}}}}{{{G}_{\text{ml}}}}{{\sum\limits_{\mathbf{s}\in \Phi_S }{\left\| {\mathbf{d}_{i}} \right\|}}^{-\alpha }}{{\left| {{g}_{{\mathbf{d}_{i}}}} \right|}^{2}}+{\bar{\sigma}^{2}}},
\end{aligned}
\label{Formula_SINR}
\end{equation}
\hrulefill
\vspace{-4mm}
\end{figure*}
where ${\bar{\sigma} }^{2} = \frac{{\sigma}^{2}}{L_{\text{pl}} P}$.
After mathematical reduction
in (\ref{Formula_SINR}), the SINR is independent of $P$, while it depends on $p_i$. 
For simplicity of presentation,
we denote the power of the intra-satellite interference signal as $I_{j,i}^{\text{intra}}=\left( \sum\limits_{m=1}^{j-1}{{{p}_{m}}}+\varpi \sum\limits_{k=j+1}^{N_{\textcolor{black}{\text{U}}}}{{{p}_{k}}} \right){{l}_{i}}^{-\alpha }{{\left| {{h}_{i}} \right|}^{2}}$, and the power of the inter-satellite interference signal as $I_{i}^{\text{inter}}=\frac{{{G}_{\text{sl}}}}{{{G}_{\text{ml}}}}{{\sum\limits_{\mathbf{s} \in \Phi_{\textcolor{black}{\text{S}}} }{\left\| {\mathbf{d}_{i}} \right\|}}^{-\alpha }}{{\left| {{g}_{{\mathbf{d}_{i}}}} \right|}^{2}}$.

The interference needs to be represented to further analyze its Laplace transform.
Denote by $r$ the distance between the projection point of the typical satellite $O$ and any of the interfering satellites. The nearest satellite to point $O$ is the typical satellite.
We use an approximation similar to that in ~\cite{liang2019non,ali2019downlink}, where typical UTs are regarded as being located at the point $O$, the center of the serving area\footnote{This can be applicable for NOMA in LEO multi-satellite networks. On the one hand, the radius of the service area is comparatively much smaller than that of the revolution radius of the satellite; on the other hand, typical UTs are uniformly distributed in the service area, so the distances from all interfering satellites to a typical UT can be approximated by the distances to the point $O$, considering the approximate mean of the distance between a typical UT and the center of the area.}. 
Hence, the distance $z$ between a typical UT and any one of the interfering satellites can be approximated by $\mathbb{E}\left\{z|r\right\}\approx r$. 
The inter-satellite interference is written as
\begin{equation}
\begin{aligned}
    I_{i}^{\text{inter}}
    & =\frac{{{G}_{\text{sl}}}}{{{G}_{\text{ml}}}}{{\sum\limits_{\mathbf{s}\in \Phi_{\textcolor{black}{\text{S}}} }{\left\| {\mathbf{d}_{i}} \right\|}}^{-\alpha }}{{\left| {{g}_{{\mathbf{d}_{i}}}} \right|}^{2}} \\
    & =\frac{{{G}_{\text{sl}}}}{{{G}_{\text{ml}}}}{{\sum\limits_{\begin{smallmatrix} 
 \mathbf{s}\in \Phi_{\textcolor{black}{\text{S}}}  \\ 
 {\mathbf{s}_{t}},t\ge 2 
\end{smallmatrix}}{\left\| {\mathbf{d}_{i}} \right\|}}^{-\alpha }}{{\left| {{g}_{{\mathbf{d}_{i}}}} \right|}^{2}}+\frac{{{G}_{\text{sl}}}}{{{G}_{\text{ml}}}}{{\sum\limits_{\begin{smallmatrix} 
 \mathbf{s}\in \Phi_{\textcolor{black}{\text{S}}}  \\ 
 {\mathbf{s}_{t}},t=1 
\end{smallmatrix}}{\left\| {\mathbf{d}_{i}} \right\|}}^{-\alpha }}{{\left| {{g}_{{\mathbf{d}_{i}}}} \right|}^{2}}.
\label{Formula_I_inter}
\end{aligned}
\end{equation}

\captionsetup{font={scriptsize}}
\begin{figure}[!t]
\begin{center}
\setlength{\abovecaptionskip}{+0.2cm}
\setlength{\belowcaptionskip}{-0.0cm}
\centering
  \includegraphics[width=5cm,height=4.8cm]{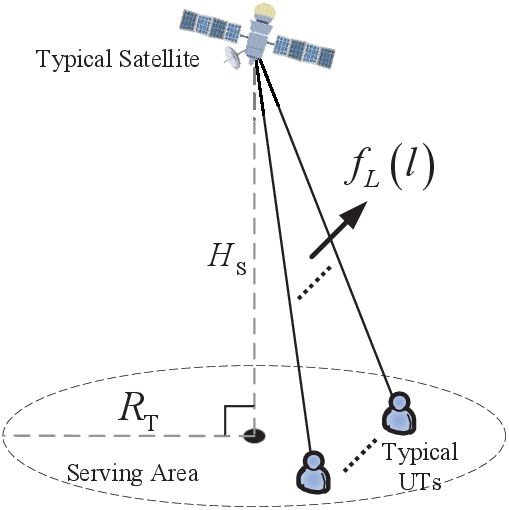}
\renewcommand\figurename{Fig.}
\caption{Distance from typical satellite to typical UTs.}
\label{fig2}
\end{center}
\vspace{-6mm}
\end{figure}

\vspace{-3mm}
\begin{lemma} 
The Laplace transform of the aggregated interference signal power is in (\ref{Formula_LT_I_inter}) at the top of the page after next.
\begin{figure*}[!ht]
\setlength{\abovecaptionskip}{-1.5cm}
\setlength{\belowcaptionskip}{-0.5cm}
\normalsize
\vspace{-2mm}
\begin{equation}
    {{\mathcal{L}}_{I_{i}^{\text{inter}}}}\left( s \right)\approx \exp \left( -\lambda_{\textcolor{black}{\text{S}}} \pi \frac{{R_{\textcolor{black}{\text{S}}}}}{{R_{\textcolor{black}{\text{E}}}}}{{\left( \frac{{{G}_{\text{sl}}}}{{{G}_{\text{ml}}}}\frac{s}{\kappa } \right)}^{\frac{2}{\alpha }}}\int_{{{\left( \frac{{{G}_{\text{sl}}}}{{{G}_{\text{ml}}}}\frac{s}{\kappa } \right)}^{-\frac{2}{\alpha }}}{{r}^{2}}}^{{{\left( \frac{{{G}_{\text{sl}}}}{{{G}_{\text{ml}}}}\frac{s}{\kappa } \right)}^{-\frac{2}{\alpha }}}{{R}_{\max }}^{2}}{\left[ 1-\frac{1}{{{\left( 1+{{u}^{-\frac{\alpha }{2}}} \right)}^{\kappa }}} \right]du} \right) {{\left( 1+\frac{{{G}_{\text{sl}}}}{{{G}_{\text{ml}}}}\frac{s}{{{r}^{\alpha }}\beta } \right)}^{-\kappa }}.
    \label{Formula_LT_I_inter}
\end{equation}
\hrulefill
\vspace{-2mm}
\end{figure*}
\end{lemma}
\vspace{-7mm}
\begin{proof}
    See Appendix A.
\end{proof}

\section{Performance Analysis}

This section provides the coverage probabilities of typical UTs based on the MSP and ISINR ordering. Considering the location randomness of the UTs and LEO satellites, related distance distributions should be first provided. Then, the expressions for coverage probabilities are derived by applying order statistics for PDFs of the link quality. Finally, a brief discussion of three types of PA schemes is shown.

\vspace{-4mm}

\subsection{Distance Distribution}
\subsubsection{Typical Satellite to Typical UT}
The circular serving area of the typical satellite is depicted in Fig. \ref{fig2}. The PDF of the link distance $l$ between the typical UT and the typical satellite is given by 
\begin{equation}
    f_{L}(l) = \begin{cases} \frac{2l}{{R_{\textcolor{black}{\text{T}}}}^{2}}, & {{L}_{\min }}\le l\le {{L}_{\max }} \\ 0, & \text{otherwise} \end{cases}
\end{equation}
where ${{L}_{\min }}={{H}_{S}}$, ${{L}_{\max }}=\sqrt{{H_{\textcolor{black}{\text{S}}}}^{2}+{R_{\textcolor{black}{\text{T}}}}^{2}}$, and $R_{\textcolor{black}{\text{T}}}$ is the radius of the serving area~\cite[Lemma 1]{shang2024multi}. The corresponding cumulative distribution function (CDF) is then calculated as
\begin{equation}
    {{F}_{L}}\left( l \right)=\int_{{{L}_{\min }}}^{l}{\frac{2t}{{R_{\textcolor{black}{\text{T}}}}^{2}}dt}=\frac{{{l}^{2}}-{H_{\textcolor{black}{\text{S}}}}^{2}}{{R_{\textcolor{black}{\text{T}}}}^{2}},    
\end{equation}
for ${{L}_{\min }}\le l\le {{L}_{\max }}$.

\subsubsection{Typical UT to interfering satellite}
Next, we present the PDF of distance $r$. 
As shown in Fig. \ref{fig3}, define a spherical cap $\mathcal{V}$ from the point $O$. The minimum distance to an interfering satellite is ${{R}_{\min }}={H_{\textcolor{black}{\text{S}}}}={R_{\textcolor{black}{\text{S}}}}-{R_{\textcolor{black}{\text{E}}}}$, while the maximum distance is ${{R}_{\max }}=\sqrt{{R_{\textcolor{black}{\text{S}}}}^{2}-{R_{\textcolor{black}{\text{E}}}}^{2}}$. 
According to~\cite[Lemma 1]{park2022tractable}, the probability of having more than one satellite, i.e., at least one interfering satellite herein, in $\mathcal{V}$ is given by
\begin{equation}
    \mathbb{P}\left[ \Phi \left( \mathcal{V} \right)>1 \right]=1-{{e}^{-\lambda_{\textcolor{black}{\text{S}}} 2\pi \left( {R_{\textcolor{black}{\text{S}}}}-{R_{\textcolor{black}{\text{E}}}} \right){R_{\textcolor{black}{\text{S}}}}}}.
    \label{Formula_prob_morethanone}
\end{equation}
Moreover, conditioned on having at least one interfering satellite, PDF of distance $r$ is~\cite[Lemma 2]{park2022tractable} 
\begin{equation}
    {{f}_{R\left| \Phi \left( \mathcal{V} \right)>1 \right.}}\left( r \right)=2\pi \lambda_{\textcolor{black}{\text{S}}} \frac{R_{\textcolor{black}{\text{S}}}}{{R_{\textcolor{black}{\text{E}}}}}\frac{{{e}^{\lambda_{\textcolor{black}{\text{S}}} \pi \frac{R_{\textcolor{black}{\text{S}}}}{{R_{\textcolor{black}{\text{E}}}}}\left( {R_{\textcolor{black}{\text{S}}}}^{2}-{R_{\textcolor{black}{\text{E}}}}^{2} \right)}}}{{{e}^{2\lambda_{\textcolor{black}{\text{S}}} \pi {R_{\textcolor{black}{\text{S}}}}\left( {R_{\textcolor{black}{\text{S}}}}-{R_{\textcolor{black}{\text{E}}}} \right)}}-1} r{{e}^{-\lambda_{\textcolor{black}{\text{S}}} \pi \frac{{{R}_{S}}}{{R_{\textcolor{black}{\text{E}}}}}{{r}^{2}}}},
    \label{Formula_PDF_r}
\end{equation}
for ${{R}_{\min }}\le r\le {{R}_{\max }}$.
Combining (\ref{Formula_prob_morethanone}) and (\ref{Formula_PDF_r}), the unconditional PDF of distance $r$ is given by (\ref{Formula_PDF_r_morethanone}) at the top of the next page.
\begin{figure*}[!ht]
\setlength{\abovecaptionskip}{-1.5cm}
\setlength{\belowcaptionskip}{-0.5cm}
\normalsize
\begin{equation}
    {{f}_{R}}\left( r \right) 
    = \begin{cases} 2\pi \lambda_{\textcolor{black}{\text{S}}} \frac{{R_{\textcolor{black}{\text{S}}}}}{{R_{\textcolor{black}{\text{E}}}}}\frac{{{e}^{\lambda_{\textcolor{black}{\text{S}}} \pi \frac{{R_{\textcolor{black}{\text{S}}}}}{{R_{\textcolor{black}{\text{E}}}}}\left( {R_{\textcolor{black}{\text{S}}}}^{2}-{R_{\textcolor{black}{\text{E}}}}^{2} \right)}}}{{{e}^{2 \lambda_{\textcolor{black}{\text{S}}} \pi {R_{\textcolor{black}{\text{S}}}}\left( {R_{\textcolor{black}{\text{S}}}}-{R_{\textcolor{black}{\text{E}}}} \right)}}-1} r{{e}^{-\lambda_{\textcolor{black}{\text{S}}} \pi \frac{{R_{\textcolor{black}{\text{S}}}}}{{R_{\textcolor{black}{\text{E}}}}}{{r}^{2}}}}\left[ 1-{{e}^{-\lambda_{\textcolor{black}{\text{S}}} 2\pi \left( {R_{\textcolor{black}{\text{S}}}}-{R_{\textcolor{black}{\text{E}}}} \right){R_{\textcolor{black}{\text{S}}}}}} \right], & {{R}_{\min }}\le r\le {{R}_{\max }}, \\ 0, & \text{otherwise}. \end{cases}
\label{Formula_PDF_r_morethanone}
\end{equation}
\hrulefill
\vspace{-4mm}
\end{figure*}

\captionsetup{font={scriptsize}}
\begin{figure}[tp]
\begin{center}
\setlength{\abovecaptionskip}{+0.2cm}
\setlength{\belowcaptionskip}{-0.0cm}
\centering
  \includegraphics[width=8.2cm,height=5.2cm]{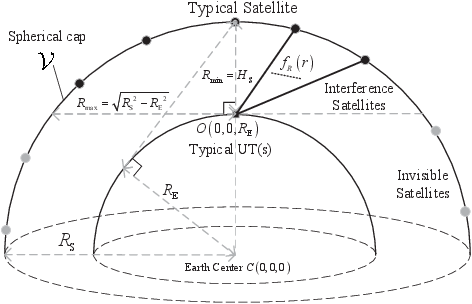}
\renewcommand\figurename{Fig.}
\caption{Distance from typical UTs to satellites.}
\label{fig3}
\end{center}
\vspace{-6mm}
\end{figure}

\subsection{Ordering-Based Coverage Probability}

For a single UT, if the SINR at UT$_i$ is higher than a target SINR threshold $\theta_i$, i.e., $\text{SINR}_i > \theta_i$ with $1\leq i\leq N_{\textcolor{black}{\text{U}}}$, the UT$_i$ is assumed to be in coverage, which is written as the event $\Lambda_i^0 = \left\{\text{SINR}_i > \theta_i \right\}$.
However, there might be more than one UT which is served simultaneously. 
To decode its intended message, the UT$_i$ needs to successfully decode the messages intended for all other UTs that are weaker than itself. For $i\leq j\leq N_{\textcolor{black}{\text{U}}}$, take $\theta_j$ for the target SINR threshold corresponding to the target rate associated with the message for UT$_j$. The joint coverage event of UT$_i$ is defined as
\begin{align}
    \Lambda_{i}
    & =\bigcap\limits_{j=i}^{N_{\textcolor{black}{\text{U}}}}{\left\{ \text{SINR}_{j}^{i}=\frac{{{p}_{j}}{{l}_{i}}^{-\alpha }{{\left| {{h}_{i}} \right|}^{2}}}{I_{j,i}^{\text{intra}}+I_{i}^{\text{inter}}+{\bar{\sigma}^{2}}}>{{\theta }_{j}} \right\}} \\
    & \overset{\left( a \right)}{\mathop{=}}\bigcap\limits_{j=i}^{N_{\textcolor{black}{\text{U}}}}{\left\{ {{\left| {{h}_{i}} \right|}^{2}}>{{l}_{i}}^{\alpha }\left( I_{i}^{\text{inter}}+{\bar{\sigma}^{2}} \right)\frac{{{\theta }_{j}}}{{{{\tilde{p}}}_{j}}} \right\}},
    \label{Fomula_CP_joint}
\end{align}
where $(a)$ is obtained via mathematical manipulations, 
and ${{\tilde{p}}_{j}}={{p}_{j}}-{{\theta }_{j}}\left( \sum\limits_{m=1}^{j-1}{{{p}_{m}}}+\varpi \sum\limits_{k=j+1}^{N_{\textcolor{black}{\text{U}}}}{{{p}_{k}}} \right)$.
To further simplify (\ref{Fomula_CP_joint}), 
denoting ${{Q}_{i}}=\underset{i\le j\le N}{\mathop{\max }}\,\frac{{{\theta }_{j}}}{{{{\tilde{p}}}_{j}}}$, then (\ref{Fomula_CP_joint}) can be written equivalently as $\Lambda_{_{i}}
    =\left\{ {{\left| {{h}_{i}} \right|}^{2}}>{{l}_{i}}^{\alpha }\left( I_{i}^{\text{inter}}+{\bar{\sigma}^{2}} \right){{Q}_{i}} \right\}$ \cite{ali2019downlink}.

According to the above derivation, the impact of intra-satellite interference is a reduction in the PA coefficient of UT$_j$. This reduction and the corresponding ${\tilde{p}}_{j}$ have no relation to the transmission rate of the message to be decoded. However, the requirement to be satisfied is ${\tilde{p}}_{j} > 0$; otherwise, the SIC mechanism cannot be successfully operated\footnote{The requirement of ${\tilde{p}}_{j} > 0$ can be regarded to as the \textit{NOMA necessary condition} for coverage\cite[Remark 3]{ali2019downlink}. When this condition is not satisfied by any one of the typical UTs, they will not be served by the typical satellite in a NOMA manner, and thus all typical UTs are assumed to be out of coverage.}. 
However, the inter-satellite interference $I_{i}^{\text{inter}}$ is affected by $G_\text{ml}$, whose influence remains unknown.
By using the distribution of the unordered link quality statistics as well as the theory of order statistics \cite{david2004order}, two ordering schemes are derived as follows.

\subsubsection{MSP Ordering}
In the MSP ordering scheme, typical UTs are ordered with regard to the ascending ordered link distance $l_i$ for $1\leq i\leq N_{\textcolor{black}{\text{U}}}$. 
With the distribution of the unordered link distance $l$ with $L_\text{min} \leq l\leq L_\text{max}$, the PDF of the ordered link distance is given by
\begin{equation}
\begin{aligned}
    {{f}_{{{L}_{i}}}}\left( l \right)
     & =\left( \begin{matrix}
       N_{\textcolor{black}{\text{U}}} - 1  \\
       i-1  \\
    \end{matrix} \right){{\left( \frac{{{l}^{2}}-{H_{\textcolor{black}{\text{S}}}}^{2}}{{R_{\textcolor{black}{\text{T}}}}^{2}} \right)}^{i-1}} \\
    & \quad \, \cdot{{\left( 1-\frac{{{l}^{2}}-{H_{\textcolor{black}{\text{S}}}}^{2}}{{R_{\textcolor{black}{\text{T}}}}^{2}} \right)}^{N_{\textcolor{black}{\text{U}}} - i}} \frac{2l}{{R_{\textcolor{black}{\text{T}}}}^{2}} N_{\textcolor{black}{\text{U}}},
\end{aligned}
\end{equation}
where in the binomial theorem $\left( \begin{matrix}
   p  \\
   q  \\
\end{matrix} \right)=\frac{p!}{q!\left( p-q \right)!}$.

\begin{theorem}
The coverage probability of the typical UT$_i$ based on MSP is approximated as (\ref{Formula_CP_MSP}) at the top of the page after next,
where ${{R}_{\min }}={{L}_{\min }}={H_{\textcolor{black}{\text{S}}}}$, ${{R}_{\max }}=\sqrt{{R_{\textcolor{black}{\text{S}}}}^{2}-{R_{\textcolor{black}{\text{E}}}}^{2}}$, and ${{L}_{\max }}=\sqrt{{H_{\textcolor{black}{\text{S}}}}^{2}+{R_{\textcolor{black}{\text{T}}}}^{2}}$.
\begin{figure*}[!ht]
\setlength{\abovecaptionskip}{-1.5cm}
\setlength{\belowcaptionskip}{-0.5cm}
\normalsize
\vspace{-2mm}
\begin{equation}
    \mathbb{P}_{\text{M}}\left( {\Lambda_{i}} \right)\approx \int_{{{L}_{\min }}}^{{{L}_{\max }}}{\int_{{{R}_{\min }}}^{{{R}_{\max }}}{\sum\limits_{k=0}^{\kappa -1}{\frac{{{\beta }^{k}}{{l}^{\alpha k}}{{Q}_{i}}^{k}}{k!}} {{\left( -1 \right)}^{k}}\frac{{{d}^{k}}{{\mathcal{L}}_{I_{i}^{\text{inter}}+{\bar{\sigma}^{2}}}}\left( s \right)}{d{{s}^{k}}}\left| _{s=\beta {{l}^{\alpha }}{{Q}_{i}}} \right. {{f}_{R}}\left( r \right)dr} {{f}_{{{L}_{i}}}}\left( l \right)dl}.
    \label{Formula_CP_MSP}
\end{equation}
\hrulefill
\end{figure*}
\end{theorem}
\begin{proof}
    Please see Appendix B.
\end{proof}

\subsubsection{ISINR Ordering}
In the ISINR ordering scheme, typical UTs are indexed in terms of the descending ordered $Z_i$ for $1\leq i\leq N_{\textcolor{black}{\text{U}}}$. The unordered $Z=\frac{{{l}^{-\alpha }}{{\left| h \right|}^{2}}}{{{I}^{\text{inter}}}+{\bar{\sigma}^{2}}}$ will be utilized similarly as in the MSP ordering scheme.
\begin{theorem}
The coverage probability of the typical UT$_i$ based on ISINR is approximated as (\ref{Formula_CP_ISINR}) at the top of the next page.
\begin{figure*}[!ht]
\setlength{\abovecaptionskip}{-1.5cm}
\setlength{\belowcaptionskip}{-0.5cm}
\normalsize
\vspace{-2mm}
\begin{equation}
    \mathbb{P}_{\text{I}}\left( {\Lambda_{i}} \right)\approx 1-\sum\limits_{k={N_{\textcolor{black}{\text{U}}}}+1-i}^{N_{\textcolor{black}{\text{U}}}}{\left( \begin{matrix}
   N_{\textcolor{black}{\text{U}}}  \\
   k  \\
\end{matrix} \right){{\left[ {{F}_{Z}}\left( {Q}_{i} \right) \right]}^{k}}{{\left[ 1-{{F}_{Z}}\left( {Q}_{i} \right) \right]}^{N_{\textcolor{black}{\text{U}}}-k}}},
\label{Formula_CP_ISINR}
\end{equation}
where
\begin{equation}
    {{F}_{Z}}\left( x \right)\approx 1-\int_{{{L}_{\min }}}^{{{L}_{\max }}}{\int_{{{R}_{\min }}}^{{{R}_{\max }}}{\sum\limits_{k=0}^{\kappa -1}{\frac{{{\left( \beta x{{l}^{\alpha }} \right)}^{k}}}{k!}}{{\left( -1 \right)}^{k}}\frac{{{d}^{k}}{{\mathcal{L}}_{{{I}^{\text{inter}}}+{\bar{\sigma}^{2}}}}\left( \beta x{{l}^{\alpha }} \right)}{d{{s}^{k}}}} {{f}_{R}}\left( r \right) {{f}_{L}}\left( l \right)drdl}.
\label{Formula_F_Z}
\end{equation}
\hrulefill
\vspace{-4mm}
\end{figure*}
\end{theorem}
\begin{proof}
    Please see Appendix C.
\end{proof}

\subsection{User Fairness-Based PA}
For a target SINR threshold $\theta_i$ and its normalized transmission rate of $\text{log}(1+\theta_i)$, the SE, data rate and sum SE of the typical UTs are expressed as 
\begin{subequations}
\begin{equation}
    \text{SE}_i = \mathbb{P}(\Lambda_i)\text{log}(1+\theta_i),
\end{equation}
\begin{equation}
    \text{R}_i = B\cdot \text{SE}_i = B\cdot \mathbb{P}(\Lambda_i)\text{log}(1+\theta_i),
\end{equation}
\begin{equation}
    \text{SE}_{\text{sum}} = \sum_{i=1}^{N_{\textcolor{black}{\text{U}}}} \text{SE}_i .
\end{equation}
\end{subequations}

\subsubsection{Equal Transmitted Power Allocation (ETPA)}
The ETPA scheme aims to ensure user fairness at the transmitter by assigning an equal share of the total transmit power to each served UT.
Specifically, with a total of $N_{\textcolor{black}{\text{U}}}$ UTs in the service area, $p_i=\frac{1}{N_{\textcolor{black}{\text{U}}}}$ of the total power $P$ will be used for each UT. This scheme reaches the goal of fairness in the initial stage with low complexity.

\subsubsection{Equal Received Power Allocation (ERPA)}
The ERPA scheme accounts for channel conditions and aims to ensure that all UTs receive equal signal power\footnote{In practical deployments, the implementation of the proposed PA algorithms, including the ERPA scheme, requires accurate knowledge of the received power at the satellite from each UT. This can be achieved through a return/uplink channel in which each UT reports its channel state information (CSI) or received signal strength indicator (RSSI) to the satellite. In the case of mobile UTs, this feedback must be performed periodically to adapt to user mobility and channel variation.

The feasibility of such a feedback mechanism is supported by existing satellite communication standards, e.g., DVB-S2X, 5G NTN, which allow for uplink signaling and measurement reporting. However, introducing a return channel may increase system overhead and latency. Although the focus of this work is on the downlink PA strategy, we recognize that a complete system-level performance evaluation should also include the cost and delay associated with CSI feedback, particularly in mobile scenarios.}.
Assume that each UT receives data at an equal power $C$, two requirements can be represented as
\begin{subequations}
\begin{equation}
    \mathrm{R}1: p_i l_0 {l_i}^{-\alpha} \mathbb{E}\left\{|h|^2\right\} = C,
\end{equation}
\begin{equation}
    \mathrm{R}2: \quad\quad\quad\quad\,\, \sum_{i=1}^{N_{\textcolor{black}{\text{U}}}}p_i = 1,
\end{equation}
\end{subequations}
where ${l_0}=L_\text{pl}$, ${l_0} {l_i}^{-\alpha}$ is the path-loss on the distance, and $|h|$ is the small-scale fading coefficient. Then, we have 
\begin{equation}
    C = \frac{1}{\sum_{i=1}^{N_{\textcolor{black}{\text{U}}}} \frac{1}{ l_0 {l_i}^{-\alpha} \mathbb{E}\left\{|h|^2\right\} }}.
\end{equation}

\subsubsection{Fixed Power Allocation (FPA)}
This paper also studies schemes with FPA coefficients. With properly selected PA coefficients, the UTs' sum SE can be maximized. However, in order to maximize the sum SE, the FPA scheme may not be as fair as the ETPA and ERPA schemes.
By comparing FPA schemes under different sets of coefficients, valuable insights will be provided for the design of NOMA in downlink LEO multi-satellite networks.

Herein, the MSP and ISINR ordering methods, together with ETPA, ERPA, and FPA, are adopted to ensure analytical tractability while reflecting representative NOMA strategies\footnote{In the literature there are many other ordering and PA schemes, each with its own optimization concern, such as ordered PA~\cite{saito2013non}, max-min rate fairness~\cite{cui2016novel}, proportional fairness (PF)~\cite{liu2016proportional}, sum-rate maximization~\cite{sun2016optimal}, etc. Although they could further enhance performance, their inclusion would require iterative optimization and significantly increase analytical complexity, which lies beyond the current scope.}.

\subsection{OMA-based Counterparts}
We consider TDMA for OMA-based LEO multi-satellite networks as a benchmark scheme. In this subsection, we briefly describe how OMA-based counterparts are measured quantitatively, following similar OMA strategies adopted in \cite[Lemma 5]{ali2019downlink} and \cite[Lemma 6]{ali2019downlink}.
First, the SINR for a typical UT$_i$ in the serving area is computed by removing the term of intra-satellite interference, i.e., removing $I_{j,i}^{\text{intra}}=\left( \sum\limits_{m=1}^{j-1}{{{p}_{m}}}+\varpi \sum\limits_{k=j+1}^{N_{\textcolor{black}{\text{U}}}}{{{p}_{k}}} \right){{l}_{i}}^{-\alpha }{{\left| {{h}_{i}} \right|}^{2}}$ in (\ref{Formula_SINR}). 
Thereby, the SINR for UT$_i$ is written as 
\begin{equation}
\begin{aligned}
    \text{SINR}_{\text{OMA}}^{i}
     & = \frac{{{l}_{i}}^{-\alpha }{{\left| {{h}_{i}} \right|}^{2}}}
     { \frac{{{G}_{\text{sl}}}}{{{G}_{\text{ml}}}}{{\sum\limits_{\mathbf{s}\in \Phi_{\textcolor{black}{\text{S}}} }{\left\| {\mathbf{d}_{i}} \right\|}}^{-\alpha }}{{\left| {{g}_{{\mathbf{d}_{i}}}} \right|}^{2}}+{\bar{\sigma}^{2}} },
\end{aligned}
\end{equation}
and the coverage at UT$_i$ is denoted as $\tilde{\Lambda}_{i}=\left\{ \text{SINR}_{\text{OMA}}^{i} > \theta_i \right\}$.
Second, consider a uniform SINR threshold $\theta_i$. By replacing $Q_i$ in (\ref{Formula_CP_MSP}) and (\ref{Formula_CP_ISINR}) with $\theta_i$, the coverage probabilities 
$\mathbb{P}_{\text{M}}\left( {\Lambda_{i}} | \theta_i \right)$ and $\mathbb{P}_{\text{I}}\left( {\Lambda_{i}} | \theta_i \right)$ 
based on the MSP ordering and ISINR ordering can be measured, respectively. 

Finally, by multiplying a time-slot coefficient $t_i$, where $\sum_{i=1}^{N_{\textcolor{black}{\text{U}}}} t_i = 1$ and $t_i=\frac{1}{N_{\textcolor{black}{\text{U}}}}$ are assumed, the sum SE of typical UTs is expressed as
\begin{equation}
    \text{SE}_{\text{sum}}^{\text{OMA}} 
    = \sum_{i=1}^{N_{\textcolor{black}{\text{U}}}} t_i \mathbb{P}_{\text{M}/\text{I}}\left( {\Lambda_{i}} | \theta_i \right) \text{log}(1+\theta_i),
\end{equation}
where $\mathbb{P}_{\text{M}/\text{I}}\left( {\Lambda_{i}} | \theta_i \right)$ is achieved by replacing $Q_i$ in (\ref{Formula_CP_MSP}) and (\ref{Formula_CP_ISINR}) with $\theta_i$.

\vspace{-3mm}

\section{Particular Cases}
This section derives some valid but closed-form expressions for the coverage probability to obtain greater tractability under particular channel conditions. 
We concentrate on the first-order channel, i.e., the Rayleigh fading channel, and the second-order channel, i.e., the non-Rayleigh fading channel with LoS components, under combinations of ordering and PA schemes.

\vspace{-3mm}

\subsection{Coverage on First-Order Channel}
\begin{corollary}
When $\kappa=1$, the coverage probability of typical UT$_i$ based on MSP is approximated as (\ref{Formula_CP_MSP_kappa_1_1}) at the top of the page.
\begin{figure*}[!ht]
\setlength{\abovecaptionskip}{-1.5cm}
\setlength{\belowcaptionskip}{-0.5cm}
\normalsize
\begin{equation}
\begin{aligned}
\mathbb{P}_{\text{M}}\left( {\Lambda_{i};\kappa=1} \right)
\approx \int_{{{L}_{\min }}}^{{{L}_{\max }}} \int_{{R}_{\min }}^{{R}_{\max }}
& \exp \left( -\lambda_{\textcolor{black}{\text{S}}} \pi \frac{R_{\textcolor{black}{\text{S}}}}{R_{\textcolor{black}{\text{E}}}}\cdot \left\{ \begin{aligned}
  & {{R}_{\max }}^{2}\left[ 1-\eta \left( s,{{R}_{\max }};\alpha ,\kappa ,{{G}_{\text{sl}}},{{G}_{\text{ml}}} \right) \right] \\ 
 & -{{r}^{2}}\left[ 1-\eta \left( s,r;\alpha ,\kappa ,{{G}_{\text{sl}}},{{G}_{\text{ml}}} \right) \right] \\ 
\end{aligned} \right\} -s\bar{\sigma}^{2} \right) \\
& \cdot {\left( 1+\frac{{{G}_{\text{sl}}}}{{{G}_{\text{ml}}}}\frac{s}{{{r}^{\alpha }}\beta } \right)}^{-\kappa } 
\left| _{s=\beta {{l}^{\alpha }}{{Q}_{i}}} \right.{{f}_{R}}\left( r \right)dr{{f}_{{{L}_{i}}}}\left( l \right)dl. 
\end{aligned}
\label{Formula_CP_MSP_kappa_1_1}
\end{equation}
\hrulefill
\vspace{-2mm}
\end{figure*}
\end{corollary}
\begin{proof}
    Please see Appendix D.
\end{proof}

\begin{corollary}
When $\kappa = 1$, the coverage probability of the typical UT$_i$ based on ISINR is approximated by (\ref{Formula_CP_ISINR_kappa_1}) at the top of the next page.
\begin{figure*}[!ht]
\setlength{\abovecaptionskip}{-1.5cm}
\setlength{\belowcaptionskip}{-0.5cm}
\normalsize
\vspace{-2mm}
\begin{equation}
    \mathbb{P}_{\text{I}} \left( {\Lambda_{i};\kappa=1} \right)
    \approx 1 - \sum\limits_{k={N_{\textcolor{black}{\text{U}}}}+1-i}^{N_{\textcolor{black}{\text{U}}}}{\left( \begin{matrix}
   N_{\textcolor{black}{\text{U}}}  \\
   k  \\
\end{matrix} \right){{\left[ {{F}_{Z}}\left( {Q}_{i};\kappa=1 \right) \right]}^{k}}{{\left[ 1-{{F}_{Z}}\left( {Q}_{i};\kappa=1 \right) \right]}^{N_{\textcolor{black}{\text{U}}}-k}}},
\label{Formula_CP_ISINR_kappa_1}
\end{equation}
where 
\begin{equation}
    {{F}_{Z}}\left( x;\kappa=1 \right)
    = 1 - \int_{{{L}_{\min }}}^{{{L}_{\max }}}{\int_{{{R}_{\min }}}^{{{R}_{\max }}}{{{\mathcal{L}}_{I_{i}^{\text{inter}}+{\bar{\sigma}^{2}}}}\left( s \right)\left| _{s=\beta x {{l}^{\alpha }}} \right. {{f}_{R}}\left( r \right)dr} {{f}_{{{L}_{i}}}}\left( l \right)dl}.
\end{equation}
\hrulefill
\end{figure*}
\end{corollary}

\begin{proof}
Taking $\kappa = 1$, the CDF of the unordered ISINR $Z$ in (\ref{Formula_F_Z}) is written as
\begin{equation}
\begin{aligned}
    &{{F}_{Z}}\left( x; \kappa=1 \right) \\
    &= 1-\int_{{{L}_{\min }}}^{{{L}_{\max }}}{\int_{{{R}_{\min }}}^{{{R}_{\max }}}{{{\mathcal{L}}_{{{I}^{\text{inter}}}+{\bar{\sigma}^{2}}}}\left( s \right)\left| _{s=\beta x{{l}^{\alpha }}} \right. {{f}_{R}}\left( r \right)dr} {{f}_{L}}\left( l \right)dl},
    \label{Formula_F_Z_kappa_1_X}
\end{aligned}
\end{equation}
where ${{\mathcal{L}}_{I_{i}^{\text{inter}}+{\bar{\sigma}^{2}}} \left( s \right)}$ and ${{\mathcal{L}}_{I_{i}^{\text{inter}}}} \left( s \right)$ are shown in (\ref{Formula_LT_I_inter_sigma}) and (\ref{Formula_LT_I_inter_closed}), respectively.
Inserting (\ref{Formula_F_Z_kappa_1_X}) into (\ref{Formula_CP_ISINR_kappa_1}) with $x=Q_i$, the expression for coverage probability based on ISINR and the first-order channel is obtained, which completes this proof.
\end{proof}

\subsection{Coverage on Second-Order Channel}
\begin{corollary}
When $\kappa = 2$, the coverage probability of the typical UT$_i$ based on MSP is approximated by (\ref{Formula_CP_MSP_kappa_2_1}) at the top of the next page.
\begin{figure*}[!ht]
\setlength{\abovecaptionskip}{-1.5cm}
\setlength{\belowcaptionskip}{-0.5cm}
\normalsize
\begin{equation}
\begin{aligned}
& \mathbb{P}_{\text{M}} \left( {{\Lambda}_{i};\kappa=2} \right) \approx \int_{{L}_{\min }}^{{L}_{\max }} \int_{R_{\min}}^{R_{\max}}  \\
&{\left\{ \begin{aligned}
  & \exp \left( -\lambda_{\textcolor{black}{\text{S}}} \pi \frac{R_{\textcolor{black}{\text{S}}}}{R_{\textcolor{black}{\text{E}}}}\cdot \left\{ \begin{aligned}
  & {{R}_{\max }}^{2}\left[ 1-\eta \left( s,{{R}_{\max }};\alpha ,\kappa ,{{G}_{\text{sl}}},{{G}_{\text{ml}}} \right) \right] \\ 
 & -{{r}^{2}}\left[ 1-\eta \left( s,r;\alpha ,\kappa ,{{G}_{\text{sl}}},{{G}_{\text{ml}}} \right) \right] \\ 
\end{aligned} \right\}-s{\bar{\sigma}^{2}} \right) \cdot \left[ \left( 1+s{\bar{\sigma}^{2}} \right) {{\left( 1+\frac{{{G}_{\text{sl}}}}{{{G}_{\text{ml}}}}\frac{s}{{{r}^{\alpha }}\beta } \right)}^{-\kappa }}  \right. \\
 & \left. - \left\{ \begin{aligned}
  & \lambda_{\textcolor{black}{\text{S}}} \pi \frac{R_{\textcolor{black}{\text{S}}}}{R_{\textcolor{black}{\text{E}}}}\frac{2}{\alpha }\cdot \left\{ \begin{aligned}
  & {{R}_{\max }}^{2}\left[ \eta \left( s,{{R}_{\max }};\alpha ,\kappa ,{{G}_{\text{sl}}},{{G}_{\text{ml}}} \right)-\frac{1}{{{\left( 1+\left( \frac{{{G}_{\text{sl}}}}{{{G}_{\text{ml}}}}\frac{s}{\kappa } \right){{R}_{\max }}^{-\alpha } \right)}^{\kappa }}} \right] \\ 
 & -{{r}^{2}}\left[ \eta \left( s,r;\alpha ,\kappa ,{{G}_{\text{sl}}},{{G}_{\text{ml}}} \right)-\frac{1}{{{\left( 1+\left( \frac{{{G}_{\text{sl}}}}{{{G}_{\text{ml}}}}\frac{s}{\kappa } \right){{r}^{-\alpha }} \right)}^{\kappa }}} \right] \\ 
\end{aligned} \right\} \\ 
 & \cdot {{\left( 1+\frac{{{G}_{\text{sl}}}}{{{G}_{\text{ml}}}}\frac{s}{{{r}^{\alpha }}\beta } \right)}^{-\kappa }}-\frac{{{G}_{\text{sl}}}}{{{G}_{\text{ml}}}}\frac{s\cdot \kappa }{{{r}^{\alpha }}\beta }{{\left( 1+\frac{{{G}_{\text{sl}}}}{{{G}_{\text{ml}}}}\frac{s}{{{r}^{\alpha }}\beta } \right)}^{-\kappa -1}} \\ 
\end{aligned} \right\} \right] \\ 
\end{aligned} \right\} } \\ 
& \left| _{s=\beta {{l}^{\alpha }}{{Q}_{i}}} \right. \cdot {{f}_{R}}\left( r \right)dr {{f}_{{{L}_{i}}}}\left( l \right)dl.
\end{aligned}
\label{Formula_CP_MSP_kappa_2_1}
\end{equation}
\hrulefill
\end{figure*}
\end{corollary}
\begin{proof}
    Please see Appendix E.
\end{proof}

\begin{corollary}
When $\kappa = 2$, the coverage probability of typical UT$_i$ based on ISINR is approximated as (\ref{Formula_CP_ISINR_kappa_2_1}) at the top of the next page, 
\begin{figure*}[!ht]
\setlength{\abovecaptionskip}{-1.5cm}
\setlength{\belowcaptionskip}{-0.5cm}
\normalsize
\vspace{-2mm}
\begin{equation}
    \mathbb{P}_{\text{I}} \left( {\Lambda_{i};\kappa=2} \right)
    \approx 1 - \sum\limits_{k=N_{\textcolor{black}{\text{U}}} + 1 - i}^{N_{\textcolor{black}{\text{U}}}}{\left( \begin{matrix}
   N_{\textcolor{black}{\text{U}}}  \\
   k  \\
\end{matrix} \right){{\left[ {{F}_{Z}}\left( Q_i;\kappa=2 \right) \right]}^{k}}{{\left[ 1-{{F}_{Z}}\left( Q_i;\kappa=2 \right) \right]}^{N_{\textcolor{black}{\text{U}}}-k}}},
\label{Formula_CP_ISINR_kappa_2_1}
\end{equation}
where
\begin{equation}
    {{F}_{Z}}\left( x; \kappa=2 \right)
    = 1-\int_{{{L}_{\min }}}^{{{L}_{\max }}}{\int_{{{R}_{\min }}}^{{{R}_{\max }}}{ \left\{ \left[ \left( 1+s{\bar{\sigma}^{2}} \right){{\mathcal{L}}_{I_{i}^{\text{inter}}}}\left( s \right)-s\frac{d{{\mathcal{L}}_{I_{i}^{\text{inter}}}}\left( s \right)}{ds} \right]{{e}^{-s{\bar{\sigma}^{2}}}} \right\} \left| _{s=\beta x{{l}^{\alpha }}} \right. {{f}_{R}}\left( r \right)dr} {{f}_{L}}\left( l \right)dl},
\label{Formula_CP_ISINR_kappa_2_1F}
\end{equation}
\hrulefill
\end{figure*}
and the expressions of ${{\mathcal{L}}_{I_{i}^{\text{inter}}}}\left( s \right)$ and $\frac{d{{\mathcal{L}}_{I_{i}^{\text{inter}}}}\left( s \right)}{ds}$ are given in (\ref{Formula_LT_I_inter_closed}) and (\ref{Formula_dLI_ds_2}), respectively.
\end{corollary}
\begin{proof}
    This proof is similar to that of Corollary 2, where we first take $\kappa=2$ for (\ref{Formula_F_Z}), and insert (\ref{Formula_LT_I_inter_closed}) and (\ref{Formula_dLI_ds_2}) with $x=Q_i$. Then, the expression for coverage probability of the typical UT$_i$ based on ISINR and second-order channel can be obtained, which completes this proof.
\end{proof}

\section{Simulations and Numerical Results}

\subsection{Simulation Setup}
This section provides a quantitative analysis of the downlink performance of LEO multi-satellite NOMA networks. 
We follow the parameter configuration in~\cite{shang2024multi}, with the radius of Earth $R_{\textcolor{black}{\text{E}}} = 6,371.393$ km, satellite altitude $H_{\textcolor{black}{\text{S}}} = 500$ km, transmit power $P=50$ dBm, main-lobe gain $G_\text{ml}=30$ dBi, side-lobe gain $G_\text{ml}=10$ dBi, noise power ${\sigma}^2=-110$ dBm.
The number of LEO satellites is set as $N_S = 600$, the radius of the serving area as $R_{\textcolor{black}{\text{T}}} = 200$ km, and the path-loss coefficient as $\alpha \approx 2.0$. The frequency band with carrier frequency $f_c = 2$ GHz and bandwidth $B=100$ MHz is considered. 
In Figs. \ref{fig4} and \ref{fig5}, PA coefficients $\left[\xi_1,\xi_2,\xi_3\right]$ for UT$_1$, UT$_2$, UT$_3$ are set $\left[0.15, 0.3, 0.55\right]$.
In the following, the pairing of an ordering scheme, e.g., MSP or ISINR, and a PA scheme, e.g., ETPA, ERPA, or FPA, is represented using a hyphen (-), such as MSP-ETPA, ISINR-ERPA.

\captionsetup{font={scriptsize}}
\begin{figure*}
\begin{center}
\centering
\setlength{\abovecaptionskip}{+0.1cm}
\setlength{\belowcaptionskip}{-0.0cm}
\centering
\includegraphics[width=7.0in, height=2.4in]{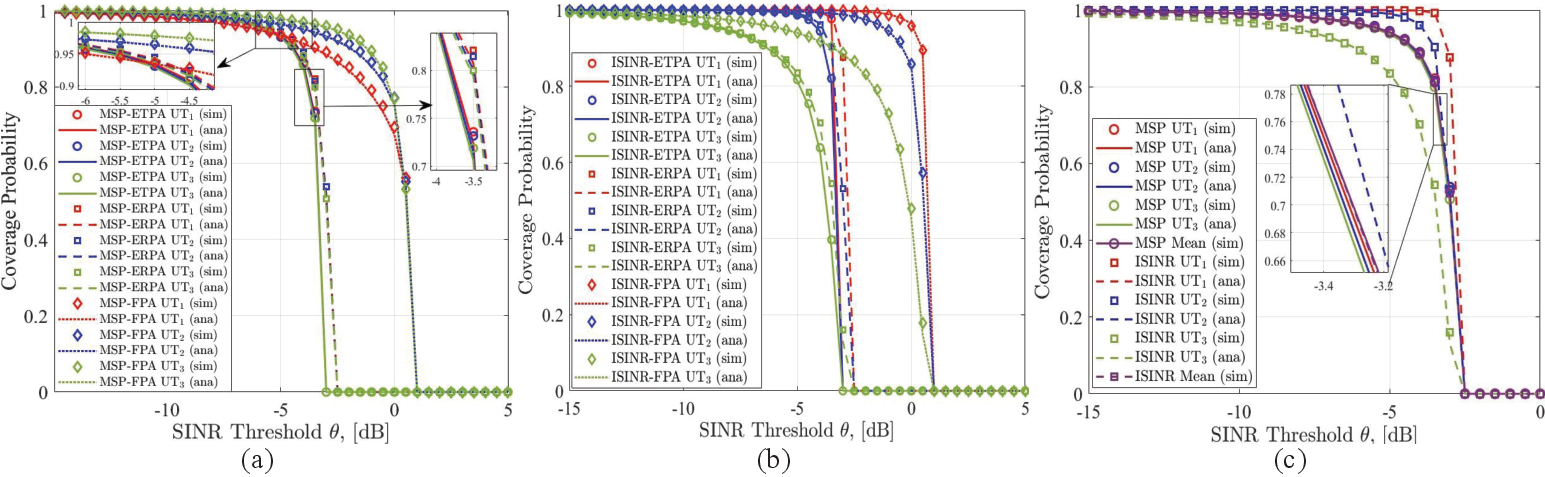}
\renewcommand\figurename{Fig.}
\caption{\scriptsize Coverage probability versus SINR threshold for $\kappa=1$. (a) MSP ordering. (b) ISINR ordering. (c) Comparison of MSP \& ISINR ordering for ERPA scheme.}
\label{fig4}
\end{center}
\vspace{-2mm}
\end{figure*}

\captionsetup{font={scriptsize}}
\begin{figure*}
\begin{center}
\centering
\setlength{\abovecaptionskip}{+0.1cm}
\setlength{\belowcaptionskip}{-0.0cm}
\centering
\includegraphics[width=7.0in, height=2.4in]{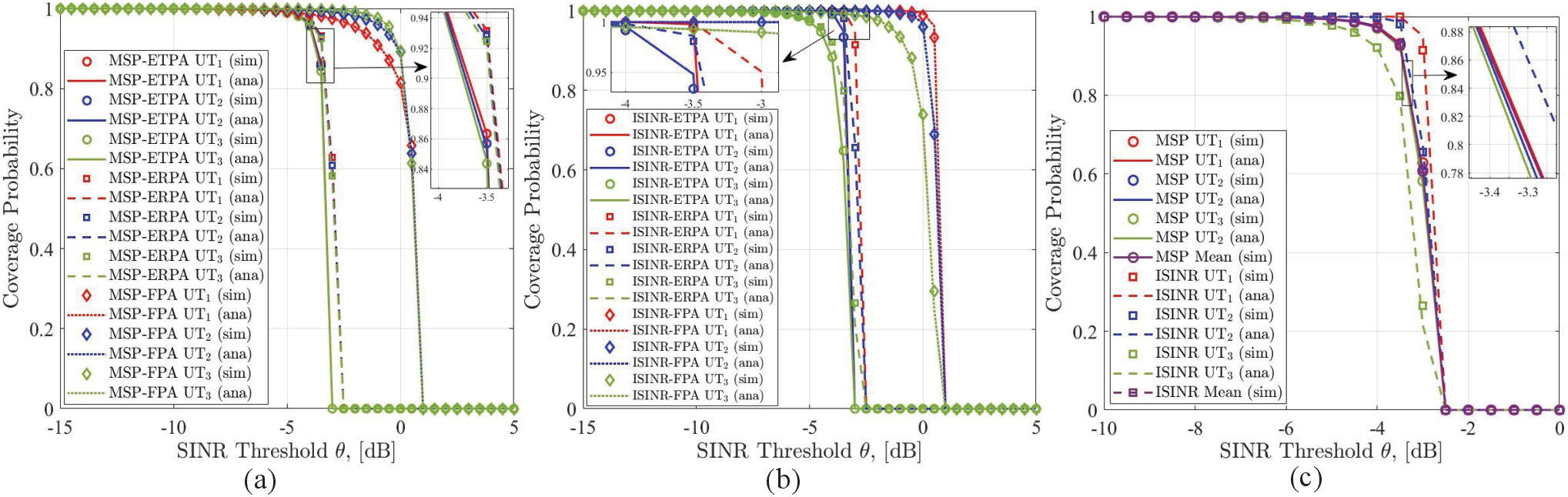}
\renewcommand\figurename{Fig.}
\caption{\scriptsize Coverage probability versus SINR threshold for $\kappa=2$. (a) MSP ordering. (b) ISINR ordering. (c) Comparison of MSP \& ISINR ordering under ERPA scheme.}
\label{fig5}
\end{center}
\vspace{-6mm}
\end{figure*}

\subsection{Coverage Probability}
In Fig. \ref{fig4}, we compare UT$_1$, UT$_2$, UT$_3$ in terms of coverage probability under different combinations of ordering schemes and PA schemes with $\kappa=1$. The precision of the expressions derived for the coverage probability when $\kappa=1$ is verified since the simulations and the analytical results match closely. Fig. \ref{fig4}(a) shows that MSP with ERPA (MSP-ERPA) scheme has a slight advantage over MSP with ETPA (MSP-ETPA) scheme. In each scheme, the performance of three UTs is almost the same, although that of UT$_1$ is a bit higher due to its better channel condition.
In Fig. \ref{fig4}(b), UTs under the ISINR-ERPA scheme have a bit higher coverage than their counterparts under the ISINR-ETPA scheme. However, the performance gap between every two UTs is enlarged - UT$_1$ and UT$_2$ have higher coverage than UT$_3$.
Fig. \ref{fig4}(c) compares the coverage of UTS under the MSP-ERPA and the ISINR-ERPA scheme. The ISINR ordering is observed to be superior for UT$_1$ and UT$_2$ while inferior for UT$_3$. This is because UT$_3$ in ISINR ordering is weaker than in MSP ordering. The mean coverage of three UTs for both orderings is almost the same.

Next, we set $\kappa=2$ to examine the effects of LoS on coverage in Fig. \ref{fig5}, where the precision of the derived expressions is also verified. In general, similar trends as in Fig. \ref{fig4} can be witnessed for all three figures. The convergence points on the x-axis are the same; however, thanks to LoS components, the coverage of $\kappa=2$ corresponding to thresholds before convergence points is higher than that of $\kappa=1$. Moreover, in both Fig. \ref{fig4} and Fig. \ref{fig5}, FPA schemes can achieve higher coverage than their ETPA and ERPA counterparts because more power allocated to weak UTs helps maintain a certain coverage performance even in poor channel conditions, although fairness is sacrificed. This also indicates that the convergence is related to the PA coefficients of different UTs and the NOMA necessary condition.

To better illustrate the effect of $\kappa$, simulations are conducted in Fig. \ref{fig6} to show the average coverage probability of the UTs under the MSP ordering and ETPA scheme. 
As mentioned earlier, the Nakagami-$m$ distribution is adopted for small-scale fading, where the factor $\kappa$ not only serves as the shape parameter of the Gamma distribution, but also corresponds to the fading parameter $m$ that characterizes the channel conditions, i.e., $\kappa = m$.
The results reveal that an increase in $\kappa$ leads to a higher average coverage probability, since a larger $\kappa$ indicates stronger LoS components in the propagation environment. This observation verifies that improving LoS conditions can effectively improve the performance of NOMA coverage in LEO satellite networks.

\captionsetup{font={scriptsize}}
\begin{figure}[tp]
\begin{center}
\setlength{\abovecaptionskip}{+0.2cm}
\setlength{\belowcaptionskip}{-0.0cm}
\centering
  \includegraphics[width=3.0in, height=2.4in]{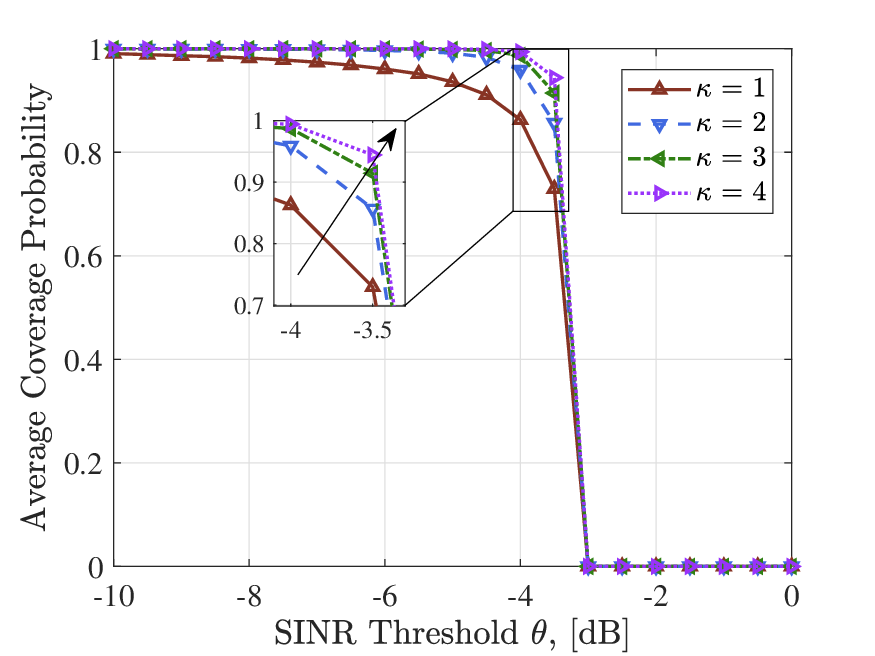}
\renewcommand\figurename{Fig.}
\caption{\textcolor{black}{Average coverage probability versus SINR threshold for different $\kappa$ under MSP ordering and ETPA scheme.}}
\label{fig6}
\end{center}
\vspace{-6mm}
\end{figure}

To enhance comparisons herein, in Fig. \ref{fig7}, we compare UTs' coverage probability under ISINR-ERPA for different constellation models, i.e. the PPP constellation model and the actual Walker-Delta, e.g. Starlink, constellation model, at an orbital altitude $H_{\textcolor{black}{\text{S}}} = 500$ km. The accuracy of the PPP model in representing real constellation setups is validated through simulations; therefore, the satellite distribution is assumed to be independent of the PPP to enable tractable analysis.
\captionsetup{font={scriptsize}}
\begin{figure}[tp]
\begin{center}
\setlength{\abovecaptionskip}{+0.2cm}
\setlength{\belowcaptionskip}{-0.0cm}
\centering
  \includegraphics[width=3.0in, height=2.4in]{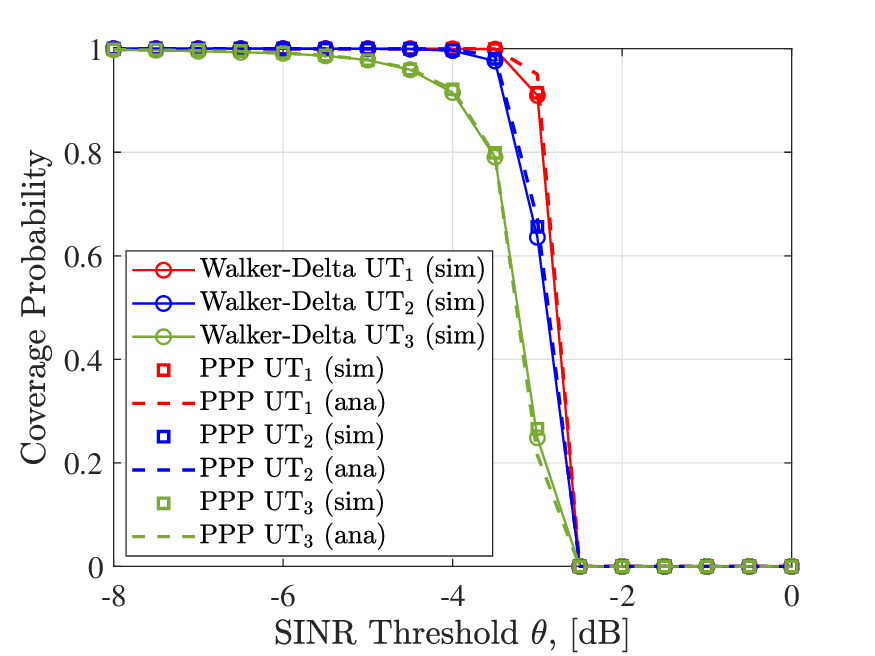}
\renewcommand\figurename{Fig.}
\caption{Coverage probability versus SINR threshold for different constellations.}
\label{fig7}
\end{center}
\vspace{-6mm}
\end{figure}

For fairness on the receiver sides and to clearly differentiate the performance of three UTs, ISINR ordering with the ERPA scheme is selected in Figs. \ref{fig8}, \ref{fig9}, and \ref{fig10}.
First, Fig. \ref{fig8} shows the influence of different SIC conditions on coverage. The performance is impaired without SIC, i.e., $\varpi=1$. However, there is virtually no difference when $\varpi$ changes from $0.01$ to $0.1$, which shows no extra benefits in appointing a perfect SIC when the SIC is relatively good.
\captionsetup{font={scriptsize}}
\begin{figure}[tp]
\begin{center}
\setlength{\abovecaptionskip}{+0.2cm}
\setlength{\belowcaptionskip}{-0.0cm}
\centering
  \includegraphics[width=3.0in, height=2.4in]{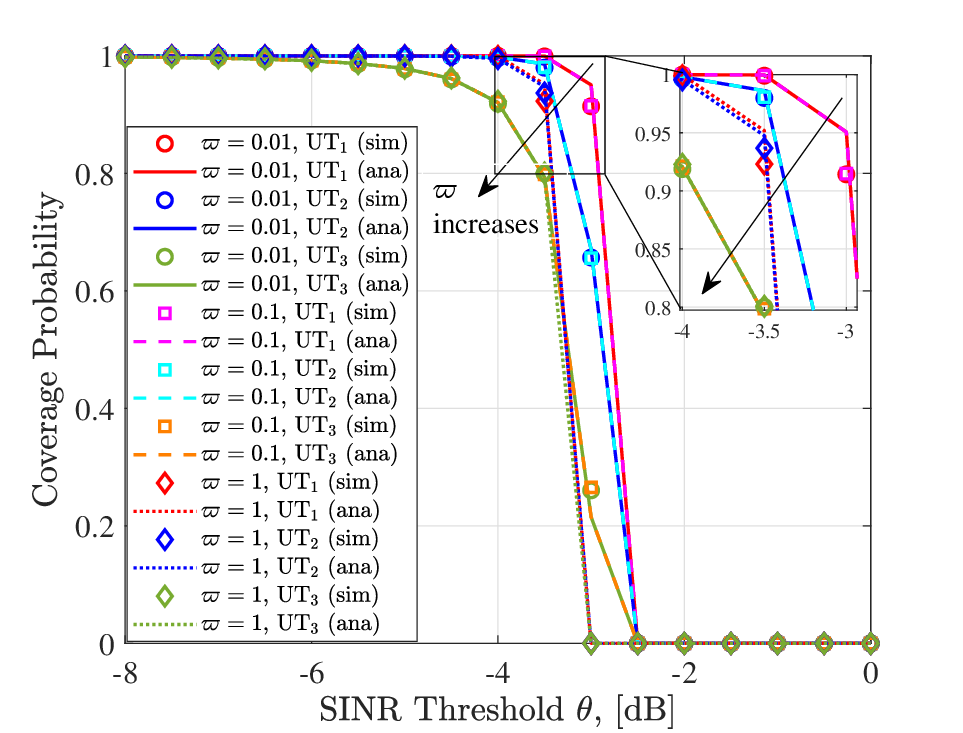}
\renewcommand\figurename{Fig.}
\caption{Coverage probability versus SINR threshold for $\kappa=2$ under different RI factors.}
\label{fig8}
\end{center}
\vspace{-6mm}
\end{figure}

Fig. \ref{fig9} presents the coverage probability of three UTs with two different satellite main-lobe gains. When increasing the main-lobe gain under a fixed side-lobe gain, the coverage of all three UTs is promoted, where the weakest UT has the highest gain, followed by UT$_2$ and UT$_1$. This accords with intuition, as enhancing main-lobe gains contributes positively to desired signals received at UTs.
\captionsetup{font={scriptsize}}
\begin{figure}[tp]
\begin{center}
\setlength{\abovecaptionskip}{+0.2cm}
\setlength{\belowcaptionskip}{-0.0cm}
\centering
  \includegraphics[width=3.0in, height=2.4in]{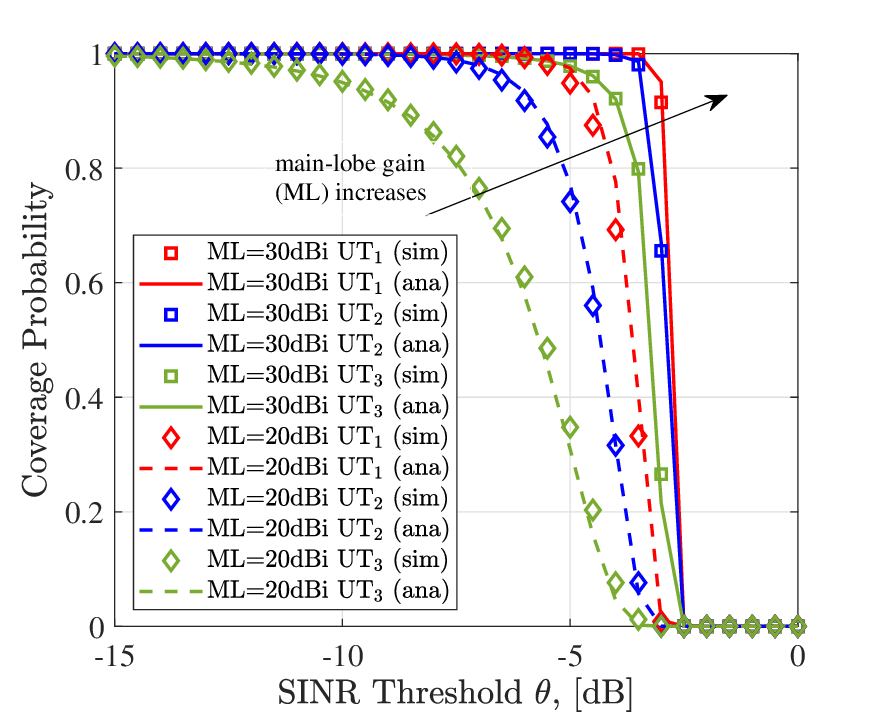}
\renewcommand\figurename{Fig.}
\caption{Coverage probability versus SINR threshold for $\kappa=2$ under different main-lobe (ML) gains.}
\label{fig9}
\end{center}
\vspace{-8mm}
\end{figure}

The number of satellites and their altitudes may have a clear impact on the system performance~\cite{shang2023coverage}. Herein, the scenario is based on the condition that the typical satellite is always the nearest satellite to the served UTs. However, the probability of occurrence of this condition can be calculated by the integral in the PDF of the nearest distance $r$ shown in (\ref{Formula_PDF_r_morethanone}) before multiplied by the coverage expressions (\ref{Formula_CP_MSP}) and (\ref{Formula_CP_ISINR}). 
In Fig. \ref{fig10}, red flags indicate the optimal number of satellites at a certain altitude. 
The conditional mean coverage probability is provided in Fig. \ref{fig10}(a), where the optimal altitude and optimal number of satellites are less than $300$ km and fewer than $1,000$ satellites, respectively. 
The unconditional mean coverage probability is shown in Fig. \ref{fig10}(b). The optimal satellite number decreases with the increase of satellite altitude.
Correspondingly, given the satellite bandwidth $B$, a plot of the mean data rate of UTs is given in \ref{fig10}(c).
This indicates that lower orbital altitudes are preferred for deploying a larger number of satellites to support the downlink LEO satellite NOMA networks.

\subsection{Spectral Efficiency}

\captionsetup{font={scriptsize}}
\begin{figure*}[tp]
\begin{center}
\setlength{\abovecaptionskip}{+0.1cm}
\setlength{\belowcaptionskip}{-0.0cm}
\centering
\includegraphics[width=7.0in, height=2.0in]{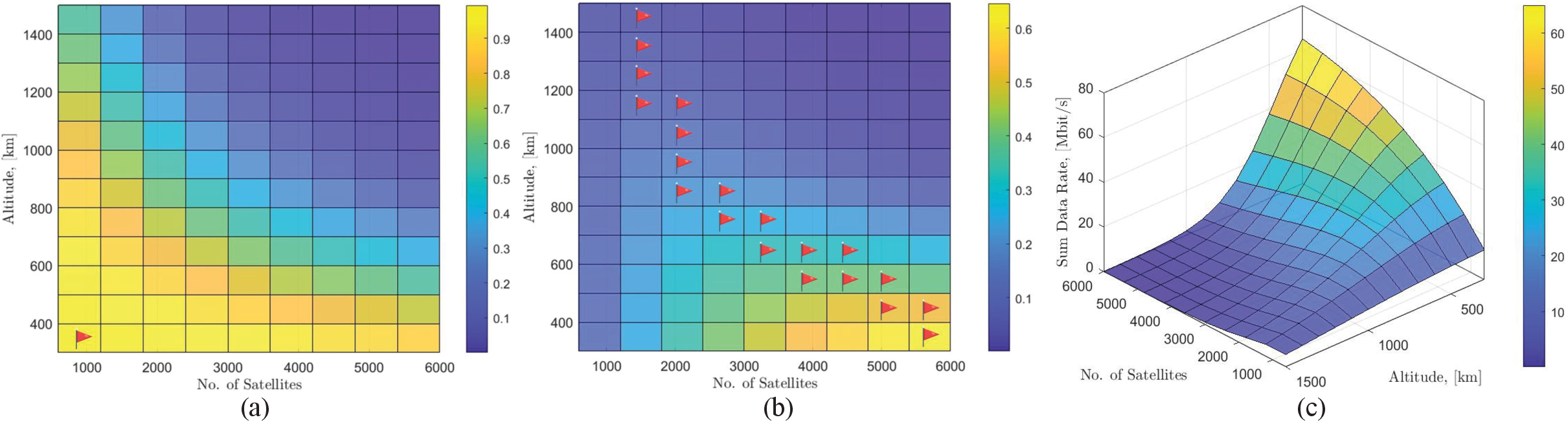}
\caption{Cases under ISINR ordering. (a) Conditional mean coverage probability. (b) Unconditional mean coverage probability. (c) Unconditional mean data rate.}
\label{fig10}
\end{center}
\vspace{-4mm}
\end{figure*}

Next, we discuss the impacts of PA coefficient sets on the sum SE. In Table \ref{tab1}, under three selected SINR thresholds, PA coefficient sets that maximize the sum SE of UTs are presented. With the increase of $\theta$, a declining trend is observed in the values of $\xi_1$ and $\xi_2$ under MSP ordering, while a growing trend is observed in the value of $\xi_3$. However, opposite trends apply for $\xi_1$, $\xi_2$ and $\xi_3$ under ISINR ordering.

\begin{table*}[!t]
\footnotesize
\caption{Optimal PA coefficient sets and maximized sum SE}
\centering
\begin{tabular}{|c|c|c|c|c|}
\hline
Ordering Scheme                 & Selected SINR Threshold $\theta$ (dB)       & -6                  & -3                & 0                   \\ \hline
\multirow{2}{*}{MSP Ordering}   & PA coefficient {[}$\xi_1$, $\xi_2$, $\xi_3${]} & {[}0.25, 0.35, 0.4{]} & {[}0.2, 0.3, 0.5{]} & {[}0.15, 0.3, 0.55{]} \\ \cline{2-5} 
                                & Sum SE (bits/s/Hz)    & 0.672182            & 1.21758           & 2.04569             \\ \hline
\multirow{2}{*}{ISINR Ordering} & PA coefficient {[}$\xi_1$, $\xi_2$, $\xi_3${]} & {[}0.1, 0.15, 0.75{]} & {[}0.1, 0.2, 0.7{]} & {[}0.15, 0.3, 0.55{]} \\ \cline{2-5} 
                                & Sum SE (bits/s/Hz)    & 0.672109            & 1.21645           & 1.86255             \\ \hline
\end{tabular}
\label{tab1}
\vspace{-2mm}
\end{table*}

The PA coefficient sets in Table \ref{tab1} aims to maximize the sum SE of the UTs instead of the fairness shown in Figs. \ref{fig4} and \ref{fig5}. Thus, we focus on the coverage performance under these PA coefficient sets. 
Fig. \ref{fig11} plots three UTs' coverage with coefficient sets FPA $\left[0.15, 0.3, 0.55\right]$, FPA2 $\left[0.2, 0.3, 0.5\right]$ and FPA3 $\left[0.1, 0.2, 0.7\right]$. We see that three sets of coverage probability lines converge to 0 at different values of $\theta$. It can be concluded that the convergence point is related to the PA coefficients of different UTs.

\captionsetup{font={scriptsize}}
\begin{figure}[tp]
\vspace{-2mm}
\begin{center}
\setlength{\abovecaptionskip}{+0.2cm}
\setlength{\belowcaptionskip}{-0.0cm}
\centering
  \includegraphics[width=3.0in, height=2.4in]{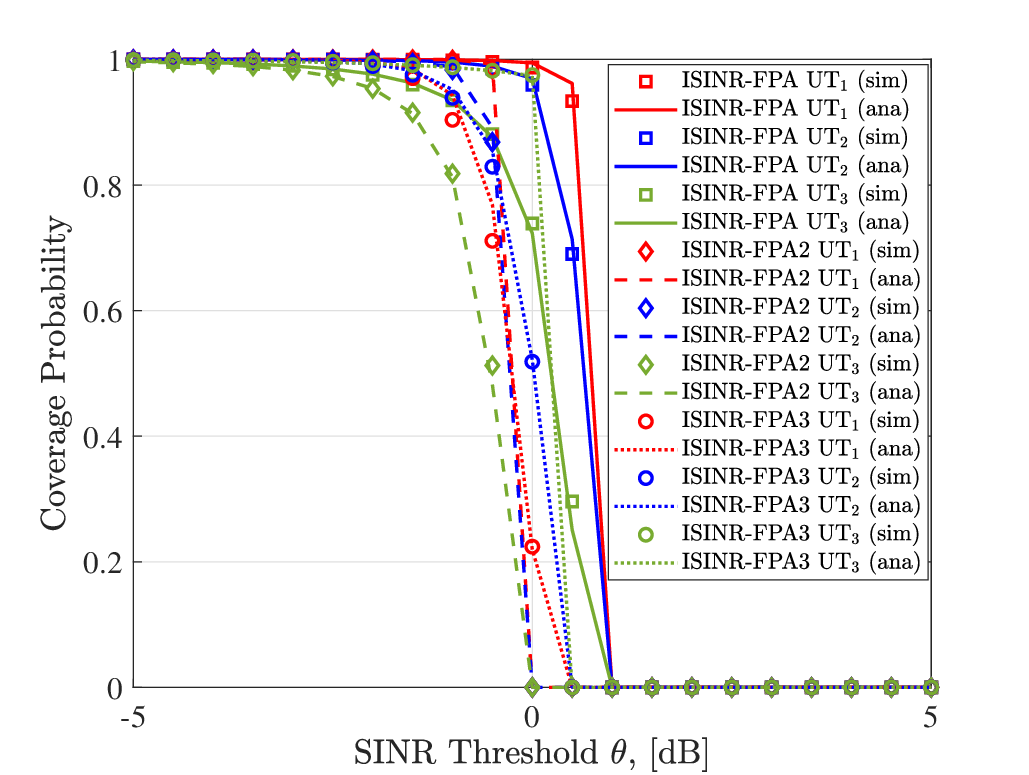}
\renewcommand\figurename{Fig.}
\caption{Comparison of three fixed PA schemes in terms of coverage probability.}
\label{fig11}
\end{center}
\vspace{-6mm}
\end{figure}

In Fig. \ref{fig12}, a comparison is made for UTs' sum SE between NOMA and OMA under both MSP and ISINR ordering schemes. 
Note that in the NOMA scheme, the plotted sum SE values represent the highest achievable values at each SINR threshold $\theta$, obtained by an ergodic search over all PA coefficient combinations for two, three, or four UTs.
In general, NOMA consistently outperforms OMA in terms of sum SE within the range of $-6$ dB $< \theta < 5$ dB.
As $\theta$ increases from $-6$ dB to $4$ dB, the sum SE gradually improves because the UTs are multiplexed.
However, due to the decoding constraints inherent in NOMA, i.e., the NOMA necessary condition as mentioned earlier, different UT numbers correspond to different SINR thresholds beyond which decoding fails, resulting in sudden performance drops.

Specifically, at $\theta = -2$ dB, the NOMA scheme achieves a local maximum sum SE of $1.8$ bits/s/Hz with four UTs. Beyond this point, the decoding for four UTs fails, and a performance drop occurs at $\theta = -1.5$ dB, where the NOMA configuration switches to three UTs. 
Following this, another performance increase is observed as $\theta$ rises, reaching the global maximum of $2.1$ bits/s/Hz around $\theta = 0.5$ or $\theta = 1$ dB. This peak corresponds to the optimal three-UT NOMA configuration.
The sum SE experiences a third increase at $\theta=4$ dB, reaching the highest $2.3$ bits/s/Hz for two UTs.
The sum SE of 2.3 bits/s/Hz is the maximum for not only the two-UT case but also the global maximum.
Compared with the maximum sum SE in OMA, an apparent gain of approximately 35\% is shown, demonstrating the effectiveness of using NOMA for downlink LEO multi-satellite networks.

Another comparison of the sum SE is made between NOMA and OMA in Fig. \ref{fig13}, with fixed SINR thresholds and an increasing number of UTs. Four combinations of schemes are involved, i.e., MSP-ETPA, MSP-ERPA, ISINR-ETPA, ISINR-ERPA. Each combination is shown to have its own optimal number of UTs for a maximum sum SE. In addition, MSP schemes have a higher sum SE than those using ISINR, and ESSPA schemes are advantageous over ETPA schemes.

\captionsetup{font={scriptsize}}
\begin{figure}[tp]
\begin{center}
\setlength{\abovecaptionskip}{+0.2cm}
\setlength{\belowcaptionskip}{-0.0cm}
\centering
  \includegraphics[width=3.0in, height=2.4in]{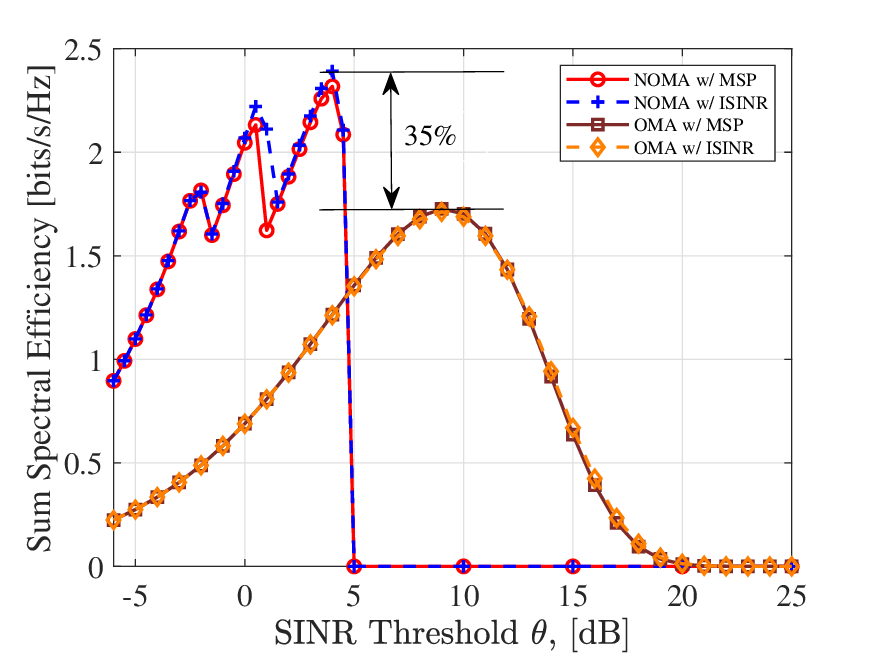}
\renewcommand\figurename{Fig.}
\caption{Comparison of NOMA and OMA under different ordering schemes.}
\label{fig12}
\end{center}
\vspace{-6mm}
\end{figure}

\captionsetup{font={scriptsize}}
\begin{figure}[tp]
\begin{center}
\setlength{\abovecaptionskip}{+0.2cm}
\setlength{\belowcaptionskip}{-0.0cm}
\centering
  \includegraphics[width=3.0in, height=2.4in]{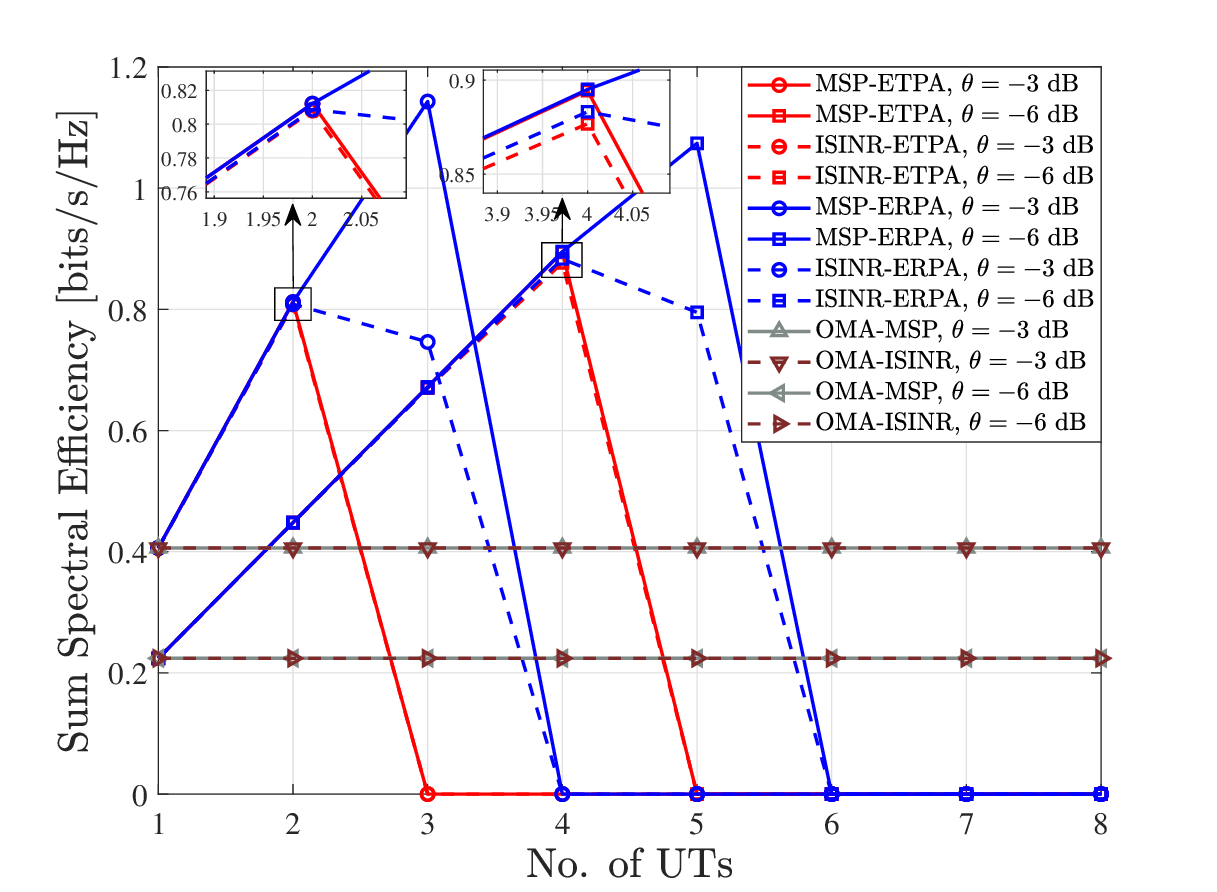}
\renewcommand\figurename{Fig.}
\caption{Comparison of NOMA and OMA with different thresholds and numbers of UTs.}
\label{fig13}
\end{center}
\vspace{-6mm}
\end{figure}

To obtain the maximum sum SE, further investigation is needed to determine whether this optimal number of UTs varies with SINR thresholds, different PA schemes, and ordering schemes.
The sub-figures in Fig. \ref{fig14} deal with this concern by involving four combinations of schemes, i.e., MSP-ETPA, MSP-ERPA, ISINR-ETPA, ISINR-ERPA. 
Different numbers of UTs are considered here to examine whether the theoretical performance can be further improved when increasing the number of multiplexed UTs, following the NOMA principle that more UTs sharing the same RB may enhance the sum spectral efficiency.
A clear similarity is that the optimal UT number for the maximum sum SE is two, with an SINR threshold around $0$ dB. Moreover, for most SINR thresholds, the sum SE increases with a growing number of UTs\footnote{The exceptions are that the sum SE drops at $\theta=-8.5$ dB with 8 UTs in MSP-ETPA and ISINR-ETPA, drops at $\theta=-7.5$ dB with 7 UTs in MSP-ERPA and ISINR-ERPA, before the sum SE reaches zero. This does not influence the overall trend.} and reaches a certain highest value before dropping to zero.

A tradeoff between UT numbers and SINR thresholds is further revealed: lower SINR thresholds and higher user densities allow more UTs to be multiplexed for improved spectral efficiency, while higher SINR thresholds or sparse scenarios favor fewer UTs.
This trend also indicates that, although adding more UTs can theoretically improve the sum SE, the gain eventually saturates due to the increased SINR threshold and potentially SIC processing burden in practice. 
Hence, while the multi-user results provide theoretical insights, practical NOMA-enabled LEO satellite implementations are expected to involve an appropriate number of UTs based on its SIC capability to balance spectral efficiency and system complexity.
Therefore, the tradeoff between the number of UTs served and the complexity of the system decoding must be considered for practical NOMA-enabled LEO satellite implementations.

\captionsetup{font={scriptsize}}
\begin{figure*}[tp]
\begin{center}
\vspace{-2mm}
\centering
\setlength{\abovecaptionskip}{+0.2cm}
\setlength{\belowcaptionskip}{-0.0cm}
\centering
    \includegraphics[width=7.1in, height=1.8in]{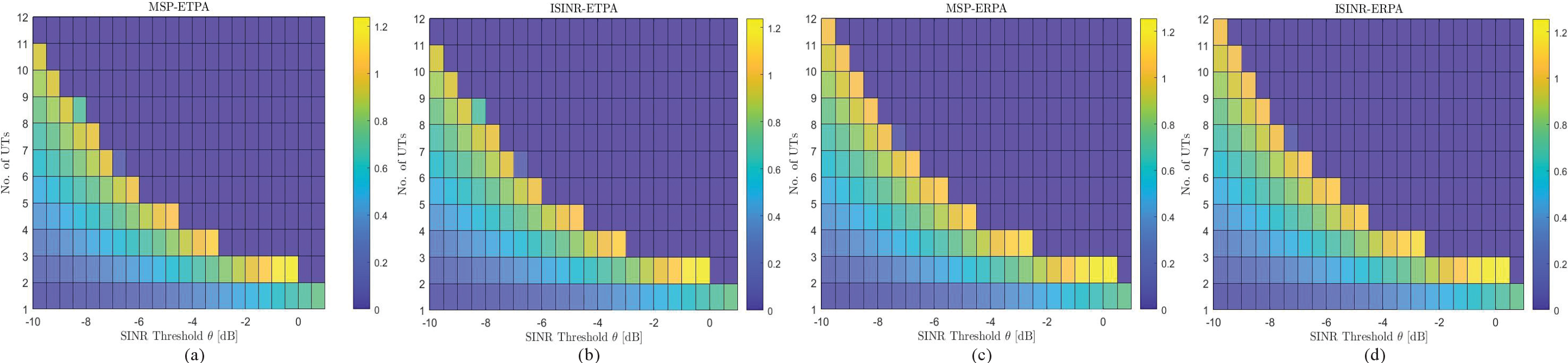}
    \caption{Sum SE vs. SINR threshold and UT number: (a) MSP ordering with ETPA. (b) ISINR ordering with ETPA. (c) MSP ordering with ERPA. (d) ISINR ordering with ERPA.}
\label{fig14}
\end{center}
\vspace{-6mm}
\end{figure*}

\section{Conclusion}
This paper developed a theoretical model to analyze downlink NOMA in LEO multi-satellite networks, considering a generalized number of UTs.
Two NOMA ordering techniques, i.e., MSP-based and ISINR-based ordering, were examined, and the corresponding coverage probabilities of served UTs were derived separately. 
Our results show that LoS components provide better coverage for served UTs compared to Rayleigh fading channels. The improvement in coverage also comes with better SIC effects of NOMA, although near-perfect SIC does not yield additional benefits. Furthermore, an increase in main-lobe gain contributes positively to overall system performance. 
There is a tradeoff between the number of satellites and their altitudes to maximize the sum SE of the UTs, with fewer satellites at higher altitudes being more advantageous. 
Additionally, while a maximum gain $35\%$ of NOMA over OMA can be obtained, there exists a maximum SINR threshold for each PA scheme. When SINR exceeds this threshold, the benefits of NOMA over OMA diminish. 
While the maximum sum SE is achieved in the two-UT case, the optimal number of UTs varies for different feasible SINR thresholds: fewer UTs are preferred at higher feasible thresholds, whereas more UTs are beneficial at lower thresholds.

With the key findings presented above, practical NOMA use cases can be envisioned in dense-area broadband access, machine-type communications-based satellite Internet of things (IoT), and air–ground integrated networks. Particularly, in densely populated areas, when system complexity and user density allow, more multiplexed NOMA UTS can be supported under a lower SINR threshold to enhance the sum SE. In sparser scenarios, fewer NOMA UTS with higher SINR thresholds are preferable for even higher system performance.
Future work may incorporate Earth curvature effects to improve large-scale coverage accuracy \cite{dong2025modeling}, and adopt more advanced ordering and PA algorithms such as PF \cite{liu2016proportional} and sum-rate maximization \cite{sun2016optimal} to improve the performance of NOMA-based LEO satellite networks, or explore their performance–complexity tradeoffs. Moreover, extending the framework to dynamic scenarios with user mobility and time-varying satellite geometry \cite{wang2022high}, as well as incorporating UAV-assisted modeling, especially for harsh or coverage-challenged environments \cite{dong2025uplink}, would further enhance its practicality and predictive capability for real-world NOMA-enabled LEO satellite networks.

\appendices 
\section{Proof of Lemma 1}
The Laplace transform of the first term in (\ref{Formula_I_inter}) is given by
\begin{align}
&{{\mathcal{L}}_{I_{i}^{\text{inter}}\left| {\mathbf{s}_{t}},t\ge 2 \right.}}\left( s \right) \nonumber \\
& =\mathbb{E}\left\{ {{e}^{-sI_{i}^{\text{inter}}}}\left| \left\| {\mathbf{s}_{t}}-{\mathbf{u}_{i}} \right\|=r,t\ge 2 \right. \right\} \nonumber \\
& \overset{\left( a \right)}{\mathop{=}}\,\exp \left( -\lambda_{\textcolor{black}{\text{S}}} \int_{v\in \mathcal{V}_{r}^{c}}{\left( 1-\mathbb{E}\left[ {{e}^{-s\frac{{{G}_{\text{sl}}}}{{{G}_{\text{ml}}}}{{\left| {{g}_{{\mathbf{d}_{i}}}} \right|}^{2}}{{v}^{-\alpha }}}} \right] \right)dv} \right) \nonumber \\ 
& \overset{\left( b \right)}{\mathop{=}}\,\exp \left( -\lambda_{\textcolor{black}{\text{S}}} \int_{v\in \mathcal{V}_{r}^{c}}{\left[ 1-\frac{1}{{{\left( 1+\frac{{{G}_{\text{sl}}}}{{{G}_{\text{ml}}}}\frac{s{{v}^{-\alpha }}}{\kappa } \right)}^{\kappa }}} \right]dv} \right) \nonumber \\
& \overset{\left( c \right)}{\mathop{=}}\,\exp \left( -2\pi \lambda_{\textcolor{black}{\text{S}}} \frac{R_{\textcolor{black}{\text{S}}}}{R_{\textcolor{black}{\text{E}}}} \int_{r}^{{{R}_{\max }}}{\left[ 1-\frac{1}{{{\left( 1+\frac{{{G}_{\text{sl}}}}{{{G}_{\text{ml}}}}\frac{s{{v}^{-\alpha }}}{\kappa } \right)}^{\kappa }}} \right]dv} \right) \nonumber \\ 
& \overset{\left( d \right)}{\mathop{=}}\,\exp \left( -\lambda_{\textcolor{black}{\text{S}}} \pi \frac{R_{\textcolor{black}{\text{S}}}}{R_{\textcolor{black}{\text{E}}}} {{\left( \frac{{{G}_{\text{sl}}}}{{{G}_{\text{ml}}}}\frac{s}{\kappa } \right)}^{\frac{2}{\alpha }}} \right. \nonumber \\
& {\quad\quad\quad\quad} \cdot \left. \int_{{{\left( \frac{{{G}_{\text{sl}}}}{{{G}_{\text{ml}}}}\frac{s}{\kappa } \right)}^{-\frac{2}{\alpha }}}{{r}^{2}}}^{{{\left( \frac{{{G}_{\text{sl}}}}{{{G}_{\text{ml}}}}\frac{s}{\kappa } \right)}^{-\frac{2}{\alpha }}}{{R}_{\max }}^{2}}{\left[ 1-\frac{1}{{{\left( 1+{{u}^{-\frac{\alpha }{2}}} \right)}^{\kappa }}} \right]du} \right),
\label{Formula_I_inter_1st}
\end{align}
where (a) comes from the probability generating functional of the PPP~\cite{park2022tractable}, (b) holds because $\left|{g}_{{\mathbf{d}_{i}}}\right|$ is a Nakagami-$m$ random variable, (c) follows from $\frac{\partial \left| {{\mathcal{V}}_{r}} \right|}{\partial r}=2 \frac{R_{\textcolor{black}{\text{S}}}}{R_{\textcolor{black}{\text{E}}}} \pi r$, while (d) is due to the change of variable $u={{\left( \frac{{{G}_{\text{sl}}}}{{{G}_{\text{ml}}}}\frac{s}{\kappa } \right)}^{-\frac{2}{\alpha }}}{{v}^{2}}$ and $du=2v{{\left( \frac{{{G}_{\text{sl}}}}{{{G}_{\text{ml}}}}\frac{s}{\kappa } \right)}^{-\frac{2}{\alpha }}}$. 
This does not contain the nearest interfering satellite, which brings the second term in (\ref{Formula_I_inter}).

As $\left|{g}_{{\mathbf{d}_{i}}}\right|$ is a Nakagami-$m$ random variable, $\left|{g}_{{\mathbf{d}_{i}}}\right|^2$ is a Gamma random variable. Then, using the moment generating function (MGF) of $\left|{g}_{{\mathbf{d}_{i}}}\right|^2$ and linear transformations of random variables, the Laplace transform of the second term in (\ref{Formula_I_inter}) is 
\begin{equation}
\begin{aligned}
    {{\mathcal{L}}_{I_{i}^{\text{inter}}\left| {\mathbf{s}_{t}},t=1 \right.}}\left( s \right)
    & =\mathbb{E}\left[ {{e}^{-sI_{i}^{\text{inter}}}}\left| \left\| {\mathbf{s}_{1}}-{\mathbf{u}_{1}} \right\|=r \right. \right] \\
    & =\mathbb{E}\left[ {\text{exp}\left({-s\frac{{{G}_{\text{sl}}}}{{{G}_{\text{ml}}}}{{\sum\limits_{\begin{smallmatrix} 
 \mathbf{s}\in \Phi_S  \\ 
 {\mathbf{s}_{t}},t=1 
\end{smallmatrix}}{r}}^{-\alpha }}{{\left| {{g}_{\mathbf{{d}_{i}}}} \right|}^{2}}}\right)} \right] \\
& ={{\mathbb{E}}_{z}}\left[ {{\left( 1+s\frac{{{G}_{\text{sl}}}}{{{G}_{\text{ml}}}}{{z}^{-\alpha }}\frac{1}{\beta } \right)}^{-\kappa }} \right].
\end{aligned}
\label{Formula_I_inter_2nd}
\end{equation}

Combining (\ref{Formula_I_inter_1st}) and (\ref{Formula_I_inter_2nd}), and taking the previous approximation of $\mathbb{E}\left\{z|r\right\}\approx r$, the final Laplace transform is
\begin{align}
{{\mathcal{L}}_{I_{i}^{\text{inter}}}}\left( s \right) 
& =\mathbb{E}\left[ {{e}^{-sI_{i}^{\text{inter}}}}\left| \left\| {\mathbf{s}_{t}}-{\mathbf{u}_{i}} \right\|=r,t=1 \right. \right] \nonumber \\
&{\quad} \cdot {{\mathbb{E}}_{z}}\left[ {{\left( 1+s\frac{{{G}_{\text{sl}}}}{{{G}_{\text{ml}}}}{{z}^{-\alpha }}\frac{1}{\beta } \right)}^{-\kappa }} \right] \nonumber \\ 
& \approx \exp \left( -\lambda_{\textcolor{black}{\text{S}}} \pi \frac{R_{\textcolor{black}{\text{S}}}}{R_{\textcolor{black}{\text{E}}}} {{\left( \frac{{{G}_{\text{sl}}}}{{{G}_{\text{ml}}}}\frac{s}{\kappa } \right)}^{\frac{2}{\alpha }}} 
\int_{{{\left( \frac{{{G}_{\text{sl}}}}{{{G}_{\text{ml}}}}\frac{s}{\kappa } \right)}^{-\frac{2}{\alpha }}}{{r}^{2}}}^{{{\left( \frac{{{G}_{\text{sl}}}}{{{G}_{\text{ml}}}}\frac{s}{\kappa } \right)}^{-\frac{2}{\alpha }}}{{R}_{\max }}^{2}}\right. \nonumber \\
& {\quad} \cdot \left.  {\left[ 1-\frac{1}{{{\left( 1+{{u}^{-\frac{\alpha }{2}}} \right)}^{\kappa }}} \right]du} \right) 
\cdot {{\left( 1+\frac{{{G}_{\text{sl}}}}{{{G}_{\text{ml}}}}\frac{s}{{{r}^{\alpha }}\beta } \right)}^{-\kappa }},
\end{align}
which completes the proof.

\section{Proof of Theorem 1}
Since $\left| {{h}_{i}} \right|$ is the Nakagami-$m$ random variable, the corresponding $\left| {{h}_{i}} \right|^2$ is a Gamma random variable~\cite{nakagami1960the}. The complementary CDF (CCDF) of $\left| {{h}_{i}} \right|^2$ is given by
$\mathbb{P}\left[ {{\left| {{h}_{i}} \right|}^{2}}\ge x \right]={{e}^{-\beta x}}\sum\limits_{k=0}^{\kappa -1}{\frac{{{\left( \beta x \right)}^{k}}}{k!}}$.
Leveraging this property, the coverage probability of the typical UT$_i$ based on MSP ordering is 
\begin{align}
 & \mathbb{P}_{\text{M}}\left( {{\Lambda}_{i}} \right) \nonumber \\
& ={{\mathbb{E}}_{l}}\Bigg[ {{\mathbb{E}}_{r}}\Bigg[ \mathbb{E}\Bigg[ {{e}^{-\beta \left( {{l}_{i}}^{\alpha }\left( I_{i}^{\text{inter}}+{\bar{\sigma}^{2}} \right){{Q}_{i}} \right)}} \nonumber \\
 & \cdot \sum\limits_{k=0}^{\kappa -1}{\frac{{{\beta }^{k}}{{\left( {{l}_{i}}^{\alpha }\left( I_{i}^{\text{inter}}+{\bar{\sigma}^{2}} \right){{Q}_{i}} \right)}^{k}}}{k!}} \Bigg] \Bigg] \cdot {{f}_{{{L}_{i}}}}\left( l \right)\left| \left\| {\mathbf{s}_{0}}-{\mathbf{u}_{i}} \right\|={{l}_{i}} \right. \Bigg] \nonumber \\ 
 & \overset{\left( a \right)}{\mathop{=}}\,{{\mathbb{E}}_{l}}\Bigg[ {{\mathbb{E}}_{r}}\Bigg[ \sum\limits_{k=0}^{\kappa -1}{\frac{{{\beta }^{k}}{{l}_{i}}^{\alpha k}{{Q}_{i}}^{k}}{k!}} \nonumber \\
 & \cdot {{\left( -1 \right)}^{k}}\frac{{{d}^{k}}{{\mathcal{L}}_{I_{i}^{\text{inter}}+{\bar{\sigma}^{2}}}}\left( s \right)}{d{{s}^{k}}}\left| _{s=\beta {{l}_{i}}^{\alpha }{{Q}_{i}}} \right. \Bigg] \cdot {{f}_{{{L}_{i}}}}\left( l \right)\left| \left\| {\mathbf{s}_{0}}-{\mathbf{u}_{i}} \right\|={{l}_{i}} \right. \Bigg] \nonumber \\ 
 & \overset{\left( b \right)}{\mathop{=}}\,\int_{{{L}_{\min }}}^{{{L}_{\max }}} \int_{{{R}_{\min }}}^{{{R}_{\max }}}\sum\limits_{k=0}^{\kappa -1}{\frac{{{\beta }^{k}}{{l}^{\alpha k}}{{Q}_{i}}^{k}}{k!}} \nonumber \\
 & \cdot {{\left( -1 \right)}^{k}}\frac{{{d}^{k}}{{\mathcal{L}}_{I_{i}^{\text{inter}}+{\bar{\sigma}^{2}}}}\left( s \right)}{d{{s}^{k}}}\left| _{s=\beta {{l}^{\alpha }}{{Q}_{i}}} \right. \cdot {{f}_{R}}\left( r \right)dr {{f}_{{{L}_{i}}}}\left( l \right)dl,
\end{align}
where (a) is obtained from the derivative property of the Laplace transform, i.e., $\mathbb{E}\left[ {{X}^{k}}{{e}^{-sX}} \right]={{\left( -1 \right)}^{k}}\frac{{{d}^{k}}{{\mathcal{L}}_{X}}\left( s \right)}{d{{s}^{k}}}$, and (b) follows from the expectation over $r$ and $l$, which completes the proof.

\section{Proof of Theorem 2}

The CDF of the unordered ISINR, $Z$, is expressed as
\begin{align}
    &{{F}_{Z}}\left( x \right) \nonumber \\
    &=\mathbb{P}\left( Z=\frac{{{l}^{-\alpha }}{{\left| h \right|}^{2}}}{{{I}^{\text{inter}}}+{\bar{\sigma}^{2}}}\le x \right) \nonumber \\
    & ={{\mathbb{E}}_{L,R,{{I}^{\text{inter}}}}}\left[ \mathbb{P}\left( {{\left| h \right|}^{2}}\le x{{l}^{\alpha }}\left( {{I}^{\text{inter}}}+{\bar{\sigma}^{2}} \right)\left| L,{{I}^{\text{inter}}} \right. \right) \right] \nonumber \\
    & \overset{\left( a \right)}{\mathop{=}} {{\mathbb{E}}_{L,R,{{I}^{\text{inter}}}}}\left[ 1-\sum\limits_{k=0}^{\kappa -1}{\frac{{{\left[ \beta \left( x{{l}^{\alpha }} \right) \right]}^{k}}}{k!}{{\left( {{I}^{\text{inter}}}+{\bar{\sigma}^{2}} \right)}^{k}}}{{e}^{-\beta x{{l}^{\alpha }}\left( {{I}^{\text{inter}}}+{\bar{\sigma}^{2}} \right)}} \right] \nonumber \\ 
    & \overset{\left( b \right)}{\mathop{=}} 1-\int_{{{L}_{\min }}}^{{{L}_{\max }}} \int_{{{R}_{\min }}}^{{{R}_{\max }}} 
    \sum\limits_{k=0}^{\kappa -1}{\frac{{{\left( \beta x{{l}^{\alpha }} \right)}^{k}}}{k!}} \nonumber \\
    & {\quad\quad\quad} \cdot {{\left( -1 \right)}^{k}}\frac{{{d}^{k}}{{\mathcal{L}}_{{{I}^{\text{inter}}}+{\bar{\sigma}^{2}}}}\left( \beta x{{l}^{\alpha }} \right)}{d{{s}^{k}}} {{f}_{R}}\left( r \right) {{f}_{L}}\left( l \right)drdl,
\end{align}
where (a) is due to the derivative property of the Laplace transform, i.e., $\mathbb{E}\left[ {{X}^{k}}{{e}^{-sX}} \right]={{\left( -1 \right)}^{k}}\frac{{{d}^{k}}{{\mathcal{L}}_{X}}\left( s \right)}{d{{s}^{k}}}$, and (b) is because of the expectation over both $l$ and $r$.
The CDF of the ordered ISINR $Z_i$ can be approximated as 
\begin{equation}
{F}_{{Z}_{i}}\left( x \right)
    \approx \sum\limits_{k = N_{\textcolor{black}{\text{U}}} + 1-i}^{N_{\textcolor{black}{\text{U}}}}{\left( \begin{matrix}
   N_{\textcolor{black}{\text{U}}}  \\
   k  \\
\end{matrix} \right){{\left[ {{F}_{Z}}\left( x \right) \right]}^{k}}{{\left[ 1-{{F}_{Z}}\left( x \right) \right]}^{N_{\textcolor{black}{\text{U}}} - k}}}.
\end{equation}
Finally, the coverage probability of typical UT$_i$ based on ISINR is written as $\mathbb{P}_{\text{I}}\left( {{\Lambda}_{i}} \right)=\mathbb{P}\left( {{Z}_{i}}>{{Q}_{i}} \right) = 1-{F}_{{Z}_{i}}\left( {Q}_{i} \right)$, which completes the proof.

\section{Proof of Corollary 1}

Taking $\kappa=1$, (\ref{Formula_CP_MSP}) can be expressed as
\begin{equation}
\begin{aligned}
    & \mathbb{P}_{\text{M}}\left( {\Lambda_{i};\kappa=1} \right) \\
    & = \int_{{{L}_{\min }}}^{{{L}_{\max }}}{\int_{{{R}_{\min }}}^{{{R}_{\max }}}{{{\mathcal{L}}_{I_{i}^{\text{inter}}+{\bar{\sigma}^{2}}}}\left( s \right)\left| _{s=\beta {{l}^{\alpha }}{{Q}_{i}}} \right. {{f}_{R}}\left( r \right)dr} {{f}_{{{L}_{i}}}}\left( l \right)dl},
    \label{Formula_CP_MSP_kappa_1}
\end{aligned}
\end{equation}
where
\begin{equation}
\begin{aligned}
    {{\mathcal{L}}_{I_{i}^{\text{inter}}+{\bar{\sigma}^{2}}}} \left( s \right) 
    = {{\mathcal{L}}_{I_{i}^{\text{inter}}}}\left( s \right){{\mathcal{L}}_{{\bar{\sigma}^{2}}}}\left( s \right) = {{\mathcal{L}}_{I_{i}^{\text{inter}}}}\left( s \right){{e}^{-s\bar{\sigma}^{2}}}.
    \label{Formula_LT_I_inter_sigma}
\end{aligned}
\end{equation}
The integral part in (\ref{Formula_LT_I_inter}), denoted as $F(u;s)$, is written as
\begin{align}
 & F(u;s) \nonumber \\
 & = 
 \int_{{{\left( \frac{{{G}_{\text{sl}}}}{{{G}_{\text{ml}}}}\frac{s}{\kappa } \right)}^{-\frac{2}{\alpha }}}{{r}^{2}}}^{{{\left( \frac{{{G}_{\text{sl}}}}{{{G}_{\text{ml}}}}\frac{s}{\kappa } \right)}^{-\frac{2}{\alpha }}}{{R}_{\max }}^{2}}{\left[ 1-\frac{1}{{{\left( 1+{{u}^{-\frac{\alpha }{2}}} \right)}^{\kappa }}} \right]du} \nonumber \\
 & = {{\left( \frac{{{G}_{\text{sl}}}}{{{G}_{\text{ml}}}}\frac{s}{\kappa } \right)}^{-\frac{2}{\alpha }}}\left( {{R}_{\max }}^{2}-{{r}^{2}} \right)-{{\left( \frac{{{G}_{\text{sl}}}}{{{G}_{\text{ml}}}}\frac{s}{\kappa } \right)}^{-\frac{2}{\alpha }}} \nonumber \\
 & {\quad} \cdot \left[ {{R}_{\max}}^{2}{\cdot } \eta \left( s,{R}_{\max};\alpha ,\kappa ,{{G}_{\text{sl}}},{{G}_{\text{ml}}} \right) \right. \nonumber \\
 & {\quad\quad\quad\quad\quad} \left.- {{r}^{2}}{{\cdot }} \eta \left( s,r;\alpha ,\kappa ,{{G}_{\text{sl}}},{{G}_{\text{ml}}} \right) \right] \nonumber \\ 
 & ={\left( \frac{{{G}_{\text{sl}}}}{{{G}_{\text{ml}}}}\frac{s}{\kappa } \right)}^{-\frac{2}{\alpha }}  \left[ {R_{\max}}^{2}\left( 1 - \eta \left( s,R_{\max};\alpha ,\kappa ,{{G}_{\text{sl}}},{{G}_{\text{ml}}} \right) \right) \right. \nonumber \\
 & {\quad\quad\quad\quad\quad\quad\quad\quad\quad} \left. - {r^2}\left( 1 - \eta \left( s,r;\alpha ,\kappa ,{{G}_{\text{sl}}},{{G}_{\text{ml}}} \right) \right) \right],
 \label{Formula_F_us}
\end{align}
where $_{2}{{F}_{1}}\left[ \cdot ,\cdot ,\cdot ,\cdot  \right]$ is the hypergeometric function, and $\eta \left( s,y;\alpha ,\kappa ,{{G}_{\text{sl}}},{{G}_{\text{ml}}} \right) 
    = {_{2}}{{F}_{1}}\left[ -\frac{2}{\alpha },\kappa ,\frac{\alpha -2}{\alpha },-{{y}^{-\alpha }}\left( \frac{{{G}_{\text{sl}}}}{{{G}_{\text{ml}}}}\frac{s}{\kappa } \right) \right]$. 
Inserting (\ref{Formula_F_us}) into (\ref{Formula_LT_I_inter}), the closed-form expression for Laplace transform of interference signal is written as
\begin{equation}
\begin{aligned}
&{{\mathcal{L}}_{I_{i}^{\text{inter}}}}\left( s \right)
= {{\left( 1+\frac{{{G}_{\text{sl}}}}{{{G}_{\text{ml}}}}\frac{s}{{{r}^{\alpha }}\beta } \right)}^{-\kappa }}\\
&\cdot \exp \left( -\lambda_{\textcolor{black}{\text{S}}} \pi \frac{R_{\textcolor{black}{\text{S}}}}{R_{\textcolor{black}{\text{E}}}} \left\{ \begin{aligned}
  & {{R}_{\max}}^{2}\left[ 1 - \eta \left( s,{R}_{\max};\alpha ,\kappa ,{{G}_{\text{sl}}},{{G}_{\text{ml}}} \right) \right] \\ 
 & -{{r}^{2}}\left[ 1 - \eta \left( s,r;\alpha ,\kappa ,{{G}_{\text{sl}}},{{G}_{\text{ml}}} \right) \right] \\ 
 \end{aligned} \right\} \right).
\end{aligned}
\label{Formula_LT_I_inter_closed}
\end{equation}
Inserting (\ref{Formula_LT_I_inter_closed}) into (\ref{Formula_CP_MSP_kappa_1}), a more tractable expression for coverage probability of MSP ordering is obtained in (\ref{Formula_CP_MSP_kappa_1_1}), which completes the proof.

\section{Proof of Corollary 3}
Taking $\kappa=2$, (\ref{Formula_CP_MSP}) can be expressed as
\begin{align}
    & \mathbb{P}_{\text{M}}\left( {\Lambda_{i};\kappa=2} \right) \nonumber \\
    & = \int_{{{L}_{\min }}}^{{{L}_{\max }}} \int_{{{R}_{\min }}}^{{{R}_{\max }}}
    \Bigg\{ \left[ \left( 1+s{\bar{\sigma}^{2}} \right){{\mathcal{L}}_{I_{i}^{\text{inter}}}}\left( s \right)-s\frac{d{{\mathcal{L}}_{I_{i}^{\text{inter}}}}\left( s \right)}{ds} \right] \nonumber \\
    & {\quad\quad\quad\quad\quad\quad\quad\quad} \cdot {{e}^{-s{\bar{\sigma}^{2}}}} \Bigg\} \left| _{s=\beta {{l}^{\alpha }}{{Q}_{i}}} \right. {{f}_{R}}\left( r \right)dr {{f}_{{{L}_{i}}}}\left( l \right)dl.
\label{Formula_CP_MSP_kappa_2_0}
\end{align}
Let us denote the following terms
\begin{equation} \nonumber
\begin{aligned}
    {{t}_{1}}\left( s \right) = - \lambda_{\textcolor{black}{\text{S}}} \pi \frac{R_{\textcolor{black}{\text{S}}}}{R_{\textcolor{black}{\text{E}}}} {{\left( \frac{{{G}_{\text{sl}}}}{{{G}_{\text{ml}}}}\frac{s}{\kappa } \right)}^{\frac{2}{\alpha }}}, {\;}
    {{t}_{2}}\left( s \right)={{\left( 1+\frac{{{G}_{\text{sl}}}}{{{G}_{\text{ml}}}}\frac{s}{{{r}^{\alpha }}\beta } \right)}^{-\kappa }},
\end{aligned}
\end{equation}
\begin{equation}\nonumber
    \frac{d{{t}_{1}}\left( s \right)}{ds} = - \lambda_{\textcolor{black}{\text{S}}} \pi \frac{R_{\textcolor{black}{\text{S}}}}{R_{\textcolor{black}{\text{E}}}} \frac{2}{\alpha }{{\left( \frac{{{G}_{\text{sl}}}}{{{G}_{\text{ml}}}}\frac{1}{\kappa } \right)}^{\frac{2}{\alpha }}}{{s}^{\frac{2}{\alpha }-1}},
\end{equation} 
\begin{equation}\nonumber
    \frac{d{{t}_{2}}\left( s \right)}{ds} = - \frac{{{G}_{\text{sl}}}}{{{G}_{\text{ml}}}}\frac{\kappa }{{{r}^{\alpha }}\beta }{{\left( 1+\frac{{{G}_{\text{sl}}}}{{{G}_{\text{ml}}}}\frac{s}{{{r}^{\alpha }}\beta } \right)}^{-\kappa -1}}.   
\end{equation}
Then, ${{\mathcal{L}}_{I_{i}^{\text{inter}}}}\left( s \right)$ can be rewritten as ${{\mathcal{L}}_{I_{i}^{\text{inter}}}}\left( s \right)=\exp \left[ {{t}_{1}}\left( s \right)\cdot F\left( u;s \right) \right]\cdot {{t}_{2}}\left( s \right)$, so that 
\begin{align}
    & \frac{d{{\mathcal{L}}_{I_{i}^{\text{inter}}}}\left( s \right)}{ds} 
     =\exp \left[ {{t}_{1}}\left( s \right)\cdot F\left( u;s \right) \right] \nonumber \\
    & {\quad} \cdot \left\{ \left[ \frac{d{{t}_{1}}\left( s \right)}{ds}F\left( u;s \right)+{{t}_{1}}\left( s \right)\frac{dF\left( u;s \right)}{ds} \right]{{t}_{2}}\left( s \right)+\frac{d{{t}_{2}}\left( s \right)}{ds} \right\}.
 \label{Formula_dLI_ds_1}
\end{align}

For the integral $F(u;s)$, denote its lower and upper limit as $a\left( s \right)={{\left( \frac{{{G}_{\text{sl}}}}{{{G}_{\text{ml}}}}\frac{s}{\kappa } \right)}^{-\frac{2}{\alpha }}}{{r}^{2}}$, and $b\left( s \right)={{\left( \frac{{{G}_{\text{sl}}}}{{{G}_{\text{ml}}}}\frac{s}{\kappa } \right)}^{-\frac{2}{\alpha }}}{{R}_{\max }}^{2}$, with derivatives 
$\frac{da\left( s \right)}{ds}=-\frac{2}{\alpha }{{\left( \frac{{{G}_{\text{sl}}}}{{{G}_{\text{ml}}}}\frac{1}{\kappa } \right)}^{-\frac{2}{\alpha }}}{{s}^{-\frac{2}{\alpha }-1}}{{r}^{2}}$, 
$\frac{db\left( s \right)}{ds}=-\frac{2}{\alpha }{{\left( \frac{{{G}_{\text{sl}}}}{{{G}_{\text{ml}}}}\frac{1}{\kappa } \right)}^{-\frac{2}{\alpha }}}{{s}^{-\frac{2}{\alpha }-1}}{{R}_{\max }}^{2}$.
Using the Leibniz integral rule, 
\begin{align}
& \frac{dF\left( u;s \right)}{ds} 
   = \left[ 1-\frac{1}{{{\left( 1+\left( \frac{{{G}_{\text{sl}}}}{{{G}_{\text{ml}}}}\frac{s}{\kappa } \right){{R}_{\max }}^{-\alpha } \right)}^{\kappa }}} \right] \nonumber \\ 
    & \cdot \left[ -\frac{2}{\alpha }{{R}_{\max }}^{2}{{\left( \frac{{{G}_{\text{sl}}}}{{{G}_{\text{ml}}}}\frac{1}{\kappa } \right)}^{-\frac{2}{\alpha }}}{{s}^{-\frac{2}{\alpha }-1}} \right] \nonumber \\
    & - \left[ 1-\frac{1}{{{\left( 1+\left( \frac{{{G}_{\text{sl}}}}{{{G}_{\text{ml}}}}\frac{s}{\kappa } \right){{r}^{-\alpha }} \right)}^{\kappa }}} \right] \left[ -\frac{2}{\alpha }{{r}^{2}}{{\left( \frac{{{G}_{\text{sl}}}}{{{G}_{\text{ml}}}}\frac{1}{\kappa } \right)}^{-\frac{2}{\alpha }}}{{s}^{-\frac{2}{\alpha }-1}} \right].
\label{Formula_d_F_us}
\end{align}
Inserting ${{t}_{1}}\left( s \right)$, $\frac{d{{t}_{1}}\left( s \right)}{ds}$, ${{t}_{2}}\left( s \right)$, $\frac{d{{t}_{2}}\left( s \right)}{ds}$, (\ref{Formula_F_us}) and (\ref{Formula_d_F_us}), (\ref{Formula_dLI_ds_1}) is then given by (\ref{Formula_dLI_ds_2}) at the top of the page.
\begin{figure*}[!ht]
\setlength{\abovecaptionskip}{-1.5cm}
\setlength{\belowcaptionskip}{-0.5cm}
\normalsize
\begin{equation}
\begin{aligned}
\frac{d{{\mathcal{L}}_{I_{i}^{\text{inter}}}}\left( s \right)}{ds}  
 = &\exp \left[ -\lambda_{\textcolor{black}{\text{S}}} \pi \frac{{{R}_{S}}}{{{R}_{E}}}\cdot \left\{ \begin{aligned}
  & {{R}_{\max}}^{2} \cdot \eta \left( s,{R}_{\max};\alpha ,\kappa ,{{G}_{\text{sl}}},{{G}_{\text{ml}}} \right)  
  -{r^2} \cdot \eta \left( s,r;\alpha ,\kappa ,{{G}_{\text{sl}}},{{G}_{\text{ml}}} \right) \\ 
\end{aligned} \right\} \right] \\
& \cdot \left\{ \begin{aligned}
  & \lambda_{\textcolor{black}{\text{S}}} \pi \frac{{{R}_{S}}}{{{R}_{E}}}\frac{2}{\alpha }{{s}^{-1}}\cdot \left\{ \begin{aligned}
  & {{R}_{\max }}^{2}\left[ \eta \left( s,{R}_{\max};\alpha ,\kappa ,{{G}_{\text{sl}}},{{G}_{\text{ml}}} \right) - \frac{1}{{{\left( 1+\left( \frac{{{G}_{\text{sl}}}}{{{G}_{\text{ml}}}}\frac{s}{\kappa } \right){{R}_{\max }}^{-\alpha } \right)}^{\kappa }}} \right] \\ 
 & -{{r}^{2}}\left[ \eta\left( s,r;\alpha ,\kappa ,{{G}_{\text{sl}}},{{G}_{\text{ml}}} \right) - \frac{1}{{{\left( 1+\left( \frac{{{G}_{\text{sl}}}}{{{G}_{\text{ml}}}}\frac{s}{\kappa } \right){{r}^{-\alpha }} \right)}^{\kappa }}} \right] \\ 
\end{aligned} \right\} \\ 
 & \cdot {{\left( 1+\frac{{{G}_{\text{sl}}}}{{{G}_{\text{ml}}}}\frac{s}{{{r}^{\alpha }}\beta } \right)}^{-\kappa }}-\frac{{{G}_{\text{sl}}}}{{{G}_{\text{ml}}}}\frac{\kappa }{{{r}^{\alpha }}\beta }{{\left( 1+\frac{{{G}_{\text{sl}}}}{{{G}_{\text{ml}}}}\frac{s}{{{r}^{\alpha }}\beta } \right)}^{-\kappa -1}} \\ 
\end{aligned} \right\}.
\end{aligned}
\label{Formula_dLI_ds_2}
\end{equation}
\hrulefill
\end{figure*}
Finally, inserting (\ref{Formula_dLI_ds_2}) into (\ref{Formula_CP_MSP_kappa_2_0}) and substitute ${{\mathcal{L}}_{I_{i}^{\text{inter}}}}$ with (\ref{Formula_LT_I_inter_closed}), the expression for coverage probability based on MSP ordering is obtained.

\bibliographystyle{IEEEtran}
\bibliography{references.bib}

\end{document}